\documentclass[11pt]{article}
\usepackage{amsmath, amssymb}
\usepackage{graphicx}
\usepackage{comment}
\usepackage{geometry}
\usepackage{setspace}
\usepackage[hidelinks]{hyperref}
\usepackage[authoryear,round]{natbib}
\usepackage{amsthm}
\usepackage{xcolor}
\usepackage[utf8]{inputenc}
\usepackage{subcaption}
\usepackage{placeins}
\usepackage{caption}     
\usepackage{booktabs}
\usepackage[english]{babel}
\usepackage{csquotes}
\MakeOuterQuote{"}

\newtheorem{proposition}{Proposition}
\newtheorem{lemma}{Lemma}
\newtheorem{definition}{Definition}
\newtheorem{assumption}{Assumption}

\newtheorem{remark}{Remark}

\title{The Augmentation Trap: AI Productivity and the Cost of Cognitive Offloading}
\author{Michael Caosun and Sinan Aral}
\date{}

\begin{document}
\maketitle

\begin{abstract}
Experimental evidence suggests that AI tools raise worker productivity, but also that sustained offloading can erode expertise. This creates a tradeoff when expertise is a complement to AI. To explore the consequences of this tradeoff, we develop a dynamic model in which a decision-maker chooses the intensity of practice-displacing AI offloading for a worker over time, trading immediate productivity against the erosion of worker skill. We decompose the tool's productivity effect into two components, one independent of worker expertise and one that scales with it. The model produces three main results. First, a decision-maker who fully anticipates skill erosion still rationally adopts AI when front-loaded gains outweigh long-run skill costs, lowering long-run productivity. The decomposition sorts deployments into five regions by their long-run effect, separating beneficial from harmful adoption. Second, the tradeoff introduces the potential for misaligned incentives to harm workers. When managers are short-termist or worker skill has external value, AI use can leave the worker worse off than with no AI, the outcome we call the augmentation trap. Third, when AI substitutes strongly enough for expertise, offloading can generate permanent divergence, with high-skill workers realizing their potential and low-skill workers deskilling. Small differences in managerial incentives can determine which path a worker takes.
\end{abstract}

\section{Introduction}

There is now substantial experimental evidence that AI tools raise worker productivity, and that the gains depend on both the task and the expertise of the user \citep{NoyZhang2023,BrynjolfssonLiRaymond2023,PengEtAl2023CopilotRCT,OtisEtAl2024Entrepreneurs,CuiEtAl2025,JuAral2025Collaborating}.\footnote{In a meta-analysis of 106 studies, \citet{VaccaroAlmaatouqMalone2024} find that human-AI synergy depends on the task. Teams underperform the best of human or AI alone on decision tasks and outperform on content creation.} There is also growing evidence that sustained AI use erodes the skill on which those gains depend. Previous waves of automation largely targeted tasks that could be codified and delegated to specialized systems. Language models differ because, as general-purpose tools, they can be applied to a broad range of cognitive tasks, and their value can scale with the expertise of the user. But the cognitive work that AI takes over can also be important for developing such expertise. If sustained AI use displaces the practice through which skill accumulates, short-run experiments will underestimate the long-run consequences \citep{BondiJohnson2026}. Using AI tools therefore introduces a tradeoff between the productivity they deliver now and the expertise that future productivity depends on.

To illustrate the intuition, let us consider the example of programming. A seasoned programmer can evaluate AI-generated code and spot mistakes, anticipate technical debt, and reject poor suggestions, whereas a novice is more likely to accept that output at face value. That expert judgment is built and maintained through the continuous practice of coding solutions and debugging failures. When looming deadlines make it rational to rely on passable AI output, programmers may stop exercising these core skills. Over months of routine use, even veterans begin missing errors they once caught easily. To put it briefly, in many domains expertise is forged by solving hard problems, not just acquiring the answer. Using AI to bypass that reasoning process appears to augment the worker's intelligence while actually automating its decline.

Several recent studies document patterns consistent with this mechanism. A year-long study of cancer specialists found that initial productivity gains from AI decision support came with a gradual dulling of expert judgment, which the authors term ``intuition rust'' \citep{EhsanEtAl2026AsymptomaticHarms}. Students who used ChatGPT for learning retained significantly less material at 45-day follow-up than those who learned without it \citep{Barcaui2025CognitiveCrutch}. Survey evidence, neuroimaging, and programming experiments find reduced cognitive engagement and weaker independent performance \citep{LeeEtAl2025CriticalThinking,KosmynaEtAl2025BrainOnChatGPT}. \citet{ShenTamkin2026SkillFormation} find that offloading drives skill loss. Participants who delegated coding tasks learned the least, while those who stayed cognitively engaged preserved their learning. \citet{Sarkar2026} documents that when working with agents, experienced developers produce outputs more aligned with their intent, while with autocompletion, less experienced workers accept more of the model's suggestions. If offloading drives skill loss even when the goal is learning, production settings where the incentive to preserve skill is weaker are unlikely to fare better.

We model how worker skill evolves as a function of practice-displacing offloading intensity when productivity gains from AI depend on worker skill. We decompose the tool's productivity effect into a skill-neutral and a skill-complementary component. The model yields three results. First, a fully informed decision-maker can rationally choose usage that ends in lower long-run output, because the productivity gains are immediate and the skill costs are gradual. Second, when the worker's horizon is long and the decision-maker's is sufficiently shorter, this steady-state loss becomes the augmentation trap, and the worker ends up worse off. Third, when AI substitutes strongly enough for expertise, the model can generate divergence: high-skill workers accumulate skill while low-skill workers deskill.

The model does not show that all AI use erodes skill. It applies specifically to contexts in which AI offloads practice that builds a skill that will remain valuable for production with AI. When AI instead teaches, provides feedback, or creates deliberate practice, the mechanism is weakened or reversed.

Our model draws on the IT productivity literature, which has established that IT returns depend on complementary human capital and organizational practices \citep{BrynjolfssonHitt2000,AralWeill2007,TambeHitt2012,Rock2019}. The central finding of that literature is that returns came from firms that restructured work around the technology. The $(\alpha,\beta)$ decomposition formalizes how production can be restructured. High $\beta$ means worker judgment remains productive in AI-assisted output, while high $\alpha$ with low $\beta$ means AI handles the work on its own regardless of who uses it. The skill atrophy associated with using AI depends on the offloading intensity $u$. The mechanism of skill atrophy is closely related to \citet{AcemogluPischke1998,AcemogluPischke1999}, which models firm incentives for investment in training. Specifically, the mechanism of worker skill atrophy can be understood as the opposite of training. Our skill dynamics are closely related to \citet{GanuthulaSingh2026}, who model skill as the balance between learning from practice and forgetting under AI assistance, and use simulations to show that sustained AI assistance can raise current performance while degrading long-run skill. Their framework treats the AI-assistance path as exogenously specified, whereas we endogenize offloading as a decision and distinguish skill-neutral from skill-complementary productivity channels.

The model also draws on the literature on automation bias, the tendency to defer to automated aids and lose the ability to perform without them, studied for decades in aviation and medical monitoring \citep{ParasuramanRiley1997,GoddardEtAl2012}. Operators follow incorrect automated advice and fail to notice problems. Neither experience nor training reliably eliminates the effect. \citet{LebovitzLifshitzAssafLevina2022} document this shift in radiology, where AI decision support changed the kind of cognitive work practitioners did on each case. \citet{DellAcquaEtAl2023} find a complementary pattern in a field experiment at a management consultancy. On tasks outside AI's capability frontier, workers with AI access who relied on incorrect AI output performed significantly worse than those without access. The IS post-adoption literature documents that system use often narrows after adoption, with workers settling into routinized use of a small feature set, and distinguishes leaner from richer ways of using a given system \citep{JaspersonEtAl2005,BurtonJonesStraub2006}. If engagement with AI output also narrows over time, the effective $\beta$ of a deployment drifts downward. Since our model holds $(\alpha,\beta)$ fixed, it may understate the problem. These literatures separately address skill complementarity, training investment distortions, and automation bias. Our contribution is to show that rational AI adoption can lower long-run productivity, and can leave workers worse off than no AI under incentive misalignment. This is particularly relevant when managerial policies mandate AI usage that forgoes parts of the job through which expertise develops.

The model also distinguishes two organizational uses of AI: performance extraction, in which the organization draws on worker expertise while reducing the practice that maintains it, and skill preservation, in which AI changes the task but leaves workers responsible for consequential judgment. The same AI system can occupy different regions of the model depending on how it is embedded in work.

The rest of the paper is organized as follows. Section~\ref{sec:model} introduces the dynamic model, classifies deployments into five regions using the $(\alpha,\beta)$ decomposition, shows that two forms of misaligned incentives (managerial short-termism and a worker skill externality) can push organizational AI deployment toward the trap region, and characterizes when deployment produces skill stratification. Section~\ref{sec:Discussion} discusses implications and identifies testable predictions. Section~\ref{sec:conclusion} concludes. All proofs are in the Appendix unless noted otherwise.

\section{Model}\label{sec:model}
\subsection{Model Setup}\label{sec:static}

Given access to an AI tool, a decision-maker chooses the intensity of practice-displacing cognitive offloading $u_{it} \in [0,1]$ for a worker to maximize discounted productivity. $u_{it}=0$ means the worker performs the relevant cognitive work unaided. $u_{it}=1$ means the worker fully delegates the practice that builds skill. Intermediate values represent workflows in which some fraction of the relevant practice is displaced. AI-assisted review, tutoring, critique, or explanation need not correspond to high $u_{it}$ if it preserves the cognitive practice that builds skill. The productivity gain from using AI can be decomposed into two components: a productivity gain independent from skill, which can be understood as tasks the AI handles independently, and a skill-dependent gain, which can be understood as the gain from tasks that benefit from additional human context. While using AI improves productivity, it also displaces practice, which is an opportunity for learning. Let $S_{it}$ denote worker $i$'s skill at time $t$, so working without AI yields output $S_{it}$.

\subsubsection*{Productivity}

Productivity combines a human output component that decreases in AI use, and an AI output component that increases in skill and in usage. This productivity effect has diminishing marginal returns, represented by $\gamma$:
\[
p(S_{it},u_{it})
=\;\underbrace{(1-u_{it})\,S_{it}}_{\text{human contribution}} \;+\;
\underbrace{\big[\alpha+\beta S_{it}-\gamma u_{it}\big]\,u_{it}}_{\text{productivity effect of AI usage}}.
\]

The first term in the AI contribution, $\alpha$, is the skill-independent productivity gain from using AI. The second term, $\beta S_{it}$, is the gain from tasks where the quality of AI output depends on the worker's judgment. Together, $\alpha$ and $\beta$ characterize the productivity of a \emph{usage practice}. There are many ways to interact with AI, each with different productivity parameters, so a particular AI tool has different parameters depending on the workflow. For example, in template-based report drafting, the model handles most of the work: $\alpha$ is high, and $\beta$ is low, so a senior partner extracts only marginally more value than a first-year associate. Client strategy work has low $\alpha$ and high $\beta$, because the model alone provides little, but a veteran consultant who knows what questions to ask can extract significant insight. For the programmer from the introduction, wholesale delegation of code generation is a high-$\alpha$, low-$\beta$ practice. When AI tools are unreliable, reviewing and correcting the output would be a practice with high $\beta$.

When $\beta > 1$, the productivity gain from AI more than compensates for the displaced human contribution, so higher-skill workers benefit more from the tool. When $\beta < 1$, AI partially substitutes for skill, narrowing the gap between high- and low-skill workers. The boundary $\beta = 1$ is the skill-neutral case: AI provides a uniform net benefit across skill levels.

The parameter $\gamma > 0$ enforces diminishing marginal returns to AI usage, because the easiest tasks are delegated first and coordination costs rise as AI handles a larger share of the workflow. The steady-state loss result does not depend on this specific functional form, and Section~\ref{app:robustness} characterizes when it arises for any smooth production function in which both AI and skill raise current output.

\subsubsection*{Skill Dynamics}

Let $S_{i0}$ denote worker $i$'s skill at time $t=0$. Following the learning-forgetting formulation of \citet{GanuthulaSingh2026}, we assume that skill evolves according to
\begin{equation}
\frac{dS_{it}}{dt}
=
\underbrace{\kappa\big(1-u_{it}\big)\big(\bar S_i - S_{it}\big)}_{\text{learning from practice}}
-
\underbrace{\kappa\,u_{it}\,S_{it}}_{\text{forgetting from offloading}},
\label{eq:skill-dynamics-full}
\end{equation}
where $\bar S_i$ is worker $i$'s maximum potential and $\kappa>0$ is a common learning/forgetting rate. These dynamics capture the idea that human practice builds skill. Holding usage fixed at $u_{it}$, skill converges to $\bar S_i(1-u_{it})$.\footnote{In the limit $\kappa\to\infty$ skill adjusts instantly, $S_{it}\to\bar S_i(1-u_{it})$, and the model reduces to a static benchmark in which productivity depends on current usage alone.}

Collecting terms, the law of motion simplifies to
\begin{equation}
\frac{dS_{it}}{dt}
=
\kappa\,\bar S_i(1-u_{it}) - \kappa S_{it}.
\label{eq:skill-dynamics-simplified}
\end{equation}
Setting $dS_{it}/dt=0$ and solving for the steady state gives $\hat S_i = \bar S_i(1-\hat u_i)$, where $\hat S_i, \hat u_i$ denote the steady-state skill and AI usage. Holding usage constant, an analyst who uses AI half the time ($u_i=0.5$) converges to half their potential. The skill dynamics themselves do not depend on $\alpha$ or $\beta$. The task parameters affect skill outcomes indirectly, by changing the optimal usage level $u^*$ through the value function. For the remainder of this section, we suppress the worker index $i$ and time subscript $t$ whenever there is no ambiguity.

\subsection{Dynamic Learning and the Long-Run Effects of AI Usage}
\label{sec:dynamic}

\subsubsection*{Productivity and the Dynamic Program}

Recall the production function from Section~\ref{sec:static}. Per-period productivity with AI is
\begin{equation}
p(S_{it},u_{it})
=
\underbrace{(1-u_{it})\,S_{it}}_{\text{human contribution}}
+
\underbrace{\big[\alpha+\beta S_{it}-\gamma u_{it}\big]\,u_{it}}_{\text{productivity effect of AI usage}}.
\label{eq:ct-prod}
\end{equation}
The decision-maker chooses a usage policy $u_{it}\in[0,1]$ to maximize discounted output:
\begin{equation}
V(S_{i0})
=
\int_0^\infty e^{-\delta t}
\Big((1-u_{it})S_{it} + \big[\alpha + \beta S_{it} - \gamma u_{it}\big]u_{it}\Big)\,dt,
\label{eq:value-fn}
\end{equation}
where $\delta>0$ is the decision-maker's discount rate. Because the immediate payoff from AI precedes the skill costs, the discount rate controls the tradeoff between short-run productivity and long-run capability. We compare the steady-state against a no-AI benchmark $\bar S/\delta$, representing the steady state of working without AI. In the no-AI state, workers  gradually accumulate skill, eventually arriving at full skill $\bar S$.

We interpret $\delta$ as the discount rate of the decision-maker. In Sections~\ref{sec:dynamic} through~\ref{sec:region-map}, we study the aligned benchmark in which the decision-maker internalizes the productivity consequences of skill loss. This benchmark can be interpreted as a worker choosing their own offloading intensity or as a firm that fully values the worker's future skill. Section~\ref{sec:misaligned-objectives} relaxes this alignment and studies cases in which the decision-maker does not bear all long-run skill costs. Unless otherwise noted, adoption is evaluated from $S_{i0}=\bar S$, so the no-AI lifetime benchmark is $\bar S/\delta$. 

The Bellman equation for this problem is
\[
\delta V(S_{it})
=
\max_{u_{it}\in[0,1]}
\Big\{ (1-u_{it})S_{it} + (\alpha + \beta S_{it} - \gamma u_{it})u_{it}
+ V'(S_{it})\big[\kappa\bar S_i(1-u_{it}) - \kappa S_{it}\big]\Big\}.
\]

\paragraph{Quadratic value function and linear usage policy.}
For policies $0<u<1$, the first-order condition is
\begin{equation}
u^*(S_{it}) = \frac{\alpha+(\beta-1) S_{it} - \kappa\bar S_i\,V'(S_{it})}{2\gamma}.
\label{eq:foc-usage}
\end{equation}
In the expression of the optimal usage policy, the $(\beta-1)$ term is the net effect of AI on the skill-dependent component, equal to $\beta S$ gained through complementarity minus $S$ displaced from human contribution. Optimal usage balances this immediate return against the value of the future skill that is eroded by heavier usage, $V'(S_{t})$. Because productivity is quadratic in usage and the skill dynamics are linear, the value function takes a specific form:

\begin{lemma}[Quadratic value and linear usage policy]\label{lem:quad-value}
Fix $\alpha,\beta,\gamma,\kappa,\delta>0$ and $\bar S>0$, and suppose the optimal usage policy is not a corner solution $(0< u < 1)$. Then there exist constants $a,b,c$ such that the value function $V(S)=aS^2+bS+c$ is quadratic and the optimal usage policy $u^*(S)=u_0+u_1 S$ is linear in skill.
\end{lemma}

The optimal policy is linear, so usage rises or falls with skill at the constant slope $u_1$. The intercept $u_0$ is the usage a worker with no skill would choose, the skill-independent gain net of the atrophy cost. The slope is the additional usage from skill complementarity against atrophy. The results that follow are comparative statics on these two numbers.

\subsection{Skill-Neutral AI and Steady-State Loss}
\label{sec:beta-zero}

Consider an AI tool that benefits a novice and a veteran equally, perhaps something like simple translation. This is the case $\beta=1$, when the skill-complementary effect exactly offsets the displaced human contribution, so that the net effect of AI on productivity does not depend on skill. By stripping out skill complementarity we can focus on the tension between immediate productivity and atrophy.

Under skill-neutral AI ($\beta=1$), the value function takes a linear form,
\begin{align}
V(S) &= bS + c, \qquad b = \frac{1}{\delta+\kappa},
\\[3pt]
u^* &= \frac{\alpha - \kappa \bar S/(\delta+\kappa)}{2\gamma}.
\label{eq:ustar-beta0}
\end{align}
Usage is positive only when the skill-neutral effect $\alpha$ exceeds an adoption threshold
$\alpha_{0} := \kappa \bar S/(\delta+\kappa)$.

When $u^*$ is substituted into the skill dynamics \eqref{eq:skill-dynamics-simplified}, we get a linear ordinary differential equation. Skill converges to a steady state $\hat S<\bar S$ when $u^*>0$. Comparing the resulting steady-state value to the no-AI benchmark $V^{\text{no-AI}}=\bar S/\delta$:

\begin{proposition}[Steady-state loss in the skill-neutral case]
\label{prop:trap-beta0}
Suppose $\beta=1$ and consider $\alpha<\alpha_{2}:=2\gamma+\kappa\bar S/(\delta+\kappa)$ so that the
optimal policy does not reach full automation. Define
\[
\alpha_{0} := \frac{\kappa\bar S}{\delta+\kappa},
\qquad
\alpha_{1} := \frac{(2\delta+\kappa)\bar S}{\delta+\kappa}.
\]
Then:
\begin{enumerate}\setlength{\itemsep}{2pt}
\item If $\alpha\le \alpha_{0}$, AI is not adopted ($u^*=0$) and $\hat{V}=\bar S/\delta$.
\item If $\alpha_{0}<\alpha< \alpha_{1}$, the optimal policy has $0<u^*<1$ and
raises current productivity but lowers long-run value below the
no-AI benchmark, $V(\hat S)<\bar S/\delta$. At $\alpha=\alpha_1$ adoption breaks even, with $V(\hat S)=\bar S/\delta$.
\item If $\alpha>\alpha_{1}$, AI adoption ($0<u^*<1$) improves both short-run
productivity and long-run value: $V(\hat S)>\bar S/\delta$.
\end{enumerate}
\end{proposition}

\begin{proof}[Proof of Proposition~\ref{prop:trap-beta0}]
Part 1 follows from the expression of the optimal AI usage policy. As the value function is linear with and without AI in this case, parts 2 and 3 follow by comparing the steady state value function.
\end{proof}

The region $\alpha\in(\alpha_{0},\alpha_{1}]$ produces steady-state loss. Adoption pays off in the short run because the productivity boost outweighs the discounted skill loss. 

Increasing $\delta$ expands the loss region,\footnote{Figures~\ref{fig:bzero}, \ref{fig:usage-skill}, and \ref{fig:trap-welfare} use $\gamma=1$ so that optimal policies are interior at moderate usage levels; Figure~\ref{fig:bzero} additionally requires $\gamma>\gamma^*$ (Remark~\ref{rem:alpha-order}). The region maps use the baseline $\gamma=0.1$.} since impatience decreases the value of steady-state output. Steady-state loss persists even with full knowledge of skill atrophy. Figure~\ref{fig:bzero} illustrates this pattern by plotting
$\Delta V\equiv V(\hat S) - \bar S/\delta$ against $\alpha$ for $\beta=1$. Adoption raises short-run productivity but lowers steady-state value in the intermediate $\alpha$ range, and only sufficiently large $\alpha$ improves long-run productivity.

\emph{Steady-state loss} compares long-run outcomes. After skill and usage converge, the AI path yields lower flow output than the no-AI benchmark. In Region~II the worker accepts the lower steady state as a rational decision because the transition surplus compensates for the reduced long-run output, so this is not a welfare loss. A worker may still prefer adoption if the early productivity gains compensate for the lower steady state. We use the term \emph{augmentation trap} for the stronger case in which they do not: the worker appears augmented in the short run, but the equilibrium usage policy erodes skill enough that the worker's lifetime welfare falls below what it would have been without adoption. Section~\ref{sec:misaligned-objectives} shows how this can occur when the decision-maker does not internalize the worker's long-run value of skill. 

\begin{proposition}[Steady-state loss under full automation]
\label{prop:auto-trap-beta0}
Suppose $\beta=1$ and consider $\alpha\ge\alpha_{2}:=2\gamma+\kappa\bar S/(\delta+\kappa)$ so that the
optimal policy reaches full automation. Define
\[
\alpha_{3} := \gamma+\bar S.
\]
Then:
\begin{enumerate}\setlength{\itemsep}{2pt}
\item If $\alpha_{2}\le \alpha<\alpha_{3}$, the optimal policy has $u^*=1$. Skill erodes to zero and AI
raises current productivity but yields a lower long-run value than the
no-AI benchmark: $V(\hat S)<\bar S/\delta$. At $\alpha=\alpha_3$, full automation breaks even: $V(\hat S)=\bar S/\delta$.
\item If $\alpha>\max\{\alpha_{2},\alpha_{3}\}$, full automation ($u^*=1$) improves both short-run flow
productivity and long-run value: $V(\hat S)>\bar S/\delta$.
\end{enumerate}
\end{proposition}

\begin{proof}[Proof of Proposition~\ref{prop:auto-trap-beta0}]
Part 1: The result that skill erodes to zero at full automation follows from the law of motion for skill acquisition once the worker stops practicing. That long-run productivity is lower holds by comparing the long-run productivity with and without AI given the conditions. Short run productivity must be higher as it is optimal for the worker to adopt AI with lower long-run productivity. Similarly, part 2 follows by comparing the steady state production function with and without AI.
\end{proof}

\begin{remark}[Ordering of the thresholds]\label{rem:alpha-order}
Let $\gamma^* := \delta\bar S/(\delta+\kappa)$. Since $\alpha_2-\alpha_1 = 2(\gamma-\gamma^*)$ and $\alpha_3-\alpha_2 = \gamma^*-\gamma$, the threshold ordering depends on whether $\gamma$ exceeds $\gamma^*$. If it does, then $\alpha_1<\alpha_2$ and $\alpha_3<\alpha_2$, so Proposition~\ref{prop:trap-beta0} has a nonempty interior gain region and full automation always improves on the no-AI benchmark. If $\gamma<\gamma^*$, part 3 of Proposition~\ref{prop:trap-beta0} is vacuous and long-run gains require full automation ($\alpha>\alpha_3$, Proposition~\ref{prop:auto-trap-beta0}).
\end{remark}
\begin{figure}[htbp]
  \centering
  \includegraphics[width=0.8\linewidth]{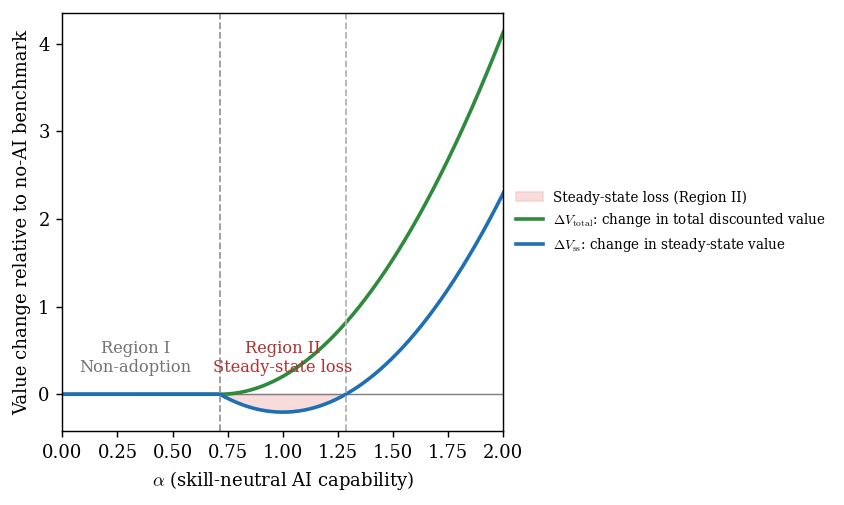}
  \caption{\textbf{Adoption surplus and steady-state loss.}
In the skill-neutral case $(\beta=1)$, the green curve plots the decision-maker's adoption surplus, $V(\bar S)-\bar S/\delta$, and the blue curve plots the steady-state value change, $V(\hat S)-\bar S/\delta$. For $\alpha_0<\alpha<\alpha_1$, adoption is privately valuable but the steady-state value is below the no-AI benchmark. Parameters: $\gamma=1$, $\kappa=0.25$, $\delta=0.1$, $\bar S=1$.}
  \label{fig:bzero}
\end{figure}

\subsection{Skill Complementarity: Usage Increasing in Skill}

The slope of the usage policy in skill determines the feedback between skill and usage. The feedback decides whether workers converge to a common steady state or separate (Section~\ref{sec:stratification}).

For $\beta>1$, AI complements skill because the skill-complementary effect more than offsets the displaced human contribution, amplifying the productivity difference between high- and low-skill workers. The optimal policy takes the linear form
\begin{equation}
u^*(S) = u_0 + u_1 S
= \frac{\alpha+(\beta-1-2\kappa a\bar S)S-\kappa b\bar S}{2\gamma},
\label{eq:ustar-beta-pos}
\end{equation}
where $a,b$ are the coefficients of the quadratic value function in Lemma~\ref{lem:quad-value} and
\[
u_0\equiv \frac{\alpha-\kappa b\bar S}{2\gamma},
\]
\[
 u_1\equiv \frac{(\beta-1-2\kappa a\bar S)}{2\gamma}
\]

Substituting \eqref{eq:ustar-beta-pos} into the skill dynamics
\eqref{eq:skill-dynamics-simplified} gives
\begin{equation}
\frac{dS}{dt}
= \kappa\big(\bar S(1-u^*(S)) - S\big)
= \hat\kappa\big(\hat S - S\big),
\label{eq:skill-dynamics-hat}
\end{equation}
with an effective learning rate and steady state
\begin{align}
\hat\kappa &= \kappa\big(1+u_1\bar S\big), \\
\hat S &= \frac{\bar S(1-u_0)}{1+u_1\bar S}. \label{eq:Shat}
\end{align}
Skill converges exponentially to $\hat S$, and usage converges to
\begin{equation}
\hat u
= u^*(\hat S)
= u_0 + u_1\hat S
= \frac{u_0+u_1\bar S}{1+u_1\bar S}.
\end{equation}

With complementarity, higher-skill workers optimally use more AI:

\begin{proposition}[Usage is increasing in skill when $\beta>1$]
\label{prop:u-increasing}
If $\beta>1$ and an interior policy is optimal, the slope of the usage policy is strictly
positive: $u_1>0$ in \eqref{eq:ustar-beta-pos}. Hence $u^*(S)$ is strictly increasing in $S$.
\end{proposition}

Two forces operate on skilled workers. Holding potential $\bar S$ fixed, a worker with higher current skill $S$ uses more AI, because complementarity makes each unit of AI usage more productive. But holding current skill fixed, a worker with higher potential $\bar S$ uses less AI, because they have further to fall from atrophy.\footnote{Remark~\ref{rem:usage-potential} in the Appendix proves this comparative static for $\beta>1$.} High-$\beta$ deployments are more likely to improve productivity, because skill-AI complementarity makes usage productive enough to discourage atrophy.

\subsection{Skill Leveling: Usage Decreasing in Skill}

When $\beta < 1$, the skill-complementary effect does not fully compensate for the displaced human contribution, so AI partially substitutes for skill. Low-skill workers now gain more from AI at the margin, so they adopt it more heavily. High-skill workers gain less and use it less. In the short run this narrows the productivity gap.

\begin{proposition}[Usage is decreasing in skill when $\beta<1$]
\label{prop:u-decreasing}
If $\beta<1$, the discriminant $D>0$,\footnote{The discriminant $D = [\kappa((\beta-1)\bar S+2\gamma)+\gamma\delta]^2 - \kappa^2(\beta-1)^2\bar S^2$ is positive when $2\kappa(1-\beta)\bar S < \gamma(\delta+2\kappa)$, that is, unless AI substitutes strongly for skill ($\beta$ well below $1$). It holds for the baseline parameterizations used in the figures below. When it fails, the interior classification no longer applies and the bounded-usage dynamics of Section~\ref{sec:stratification} take over.} and an interior policy is optimal, the slope of the usage policy is strictly
negative: $u_1<0$ in \eqref{eq:ustar-beta-pos}. Hence $u^*(S)$ is strictly decreasing in $S$.
\end{proposition}

\begin{figure}[t!]
\centering
\includegraphics[width=0.85\linewidth]{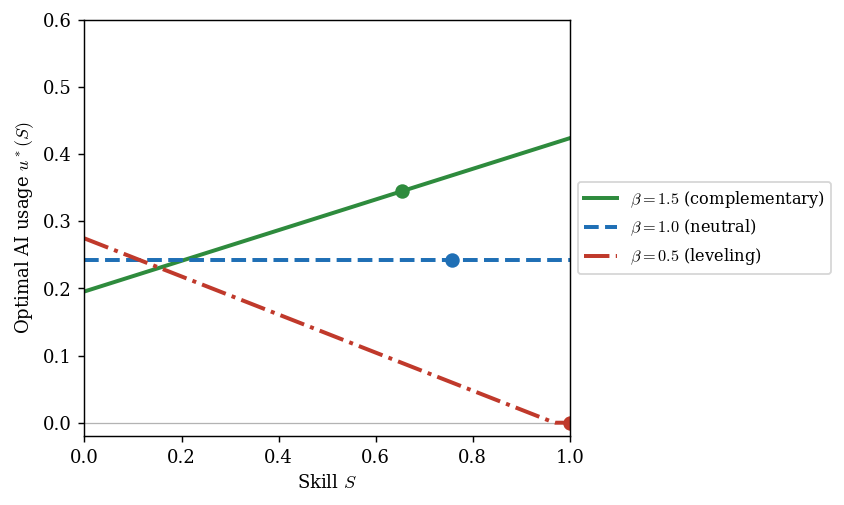}
\caption{\textbf{Optimal AI usage as a function of skill for the three complementarity regimes.}
When $\beta>1$, higher-skill workers use AI more; when $\beta<1$, lower-skill workers use AI more;
when $\beta=1$, usage is flat. Dots mark steady states. Parameters: $\alpha=1.2$, $\gamma=1.0$, $\kappa=0.25$, $\delta=0.1$, $\bar S=1$. Curves shown for $\beta=1.5$ (complementary), $\beta=1.0$ (neutral), and $\beta=0.5$ (leveling).}
\label{fig:usage-skill}
\end{figure}

Figure~\ref{fig:usage-skill} summarizes the three regimes. The sign of $u_1$ determines how the skill dynamics unfold. Two forces compete. Skill reverts toward its steady state, which pulls workers together, while the feedback between skill and usage can push them apart. When $\beta>1$, the feedback reinforces this reversion. A high-skill worker who uses AI heavily loses skill, so their usage falls and skill recovers. Similarly, a low-skill worker who avoids AI builds skill, which then raises their usage. These dynamics imply that all workers converge to a common steady state.

When $\beta<1$, the feedback instead opposes reversion. When the feedback dominates, low-skill workers are driven to full automation and high-skill workers to none, and the workforce splits. We expand on this in Section~\ref{sec:stratification}.

\FloatBarrier
\subsection{Steady-State Region Map}\label{sec:region-map}

In steady state, skill and usage settle at $(\hat S,\hat u)$, with productivity
\begin{equation}
y^\star
= \bar S(1-\hat u)^2 + \alpha\hat u - \gamma\hat u^2
+ \beta\bar S(1-\hat u)\hat u.
\label{eq:y-steady}
\end{equation}

Define
\[
\Delta y^{\mathrm{ss}} \;\equiv\; y^\star - \bar S
\]
as the steady-state output effect of the deployment. Because the policy holds skill fixed at a steady state, the value there is $y^\star/\delta$, so $\Delta y^{\mathrm{ss}}$ also signs the long-run value comparison $V(\hat S)-\bar S/\delta$. The locus $B:=\{(\alpha,\beta):\Delta y^{\mathrm{ss}}=0\}$ partitions the parameter
space into long-run output-gain and output-loss regions.

Figure~\ref{fig:nomathheatmap} partitions the parameter space into five regions according to the tool's long-run effect on productivity, separated by four boundaries:
\begin{itemize}\setlength{\itemsep}{2pt}
\item $C_0:=\{(\alpha,\beta):\hat u=0\}$, below which non-adoption is optimal;
\item $C_1:=\{(\alpha,\beta):\hat u=1\}$, above which full automation is optimal;
\item $B:=\{(\alpha,\beta):\Delta y^{\mathrm{ss}}=0\}$, where long-run output breaks even;
\item $F:=\{(\alpha,\beta):\alpha - \gamma = \bar S\}$, to the right of which full automation beats the no-AI baseline.
\end{itemize}

\begin{figure}[t!]
\centering
\includegraphics[width=0.75\linewidth]{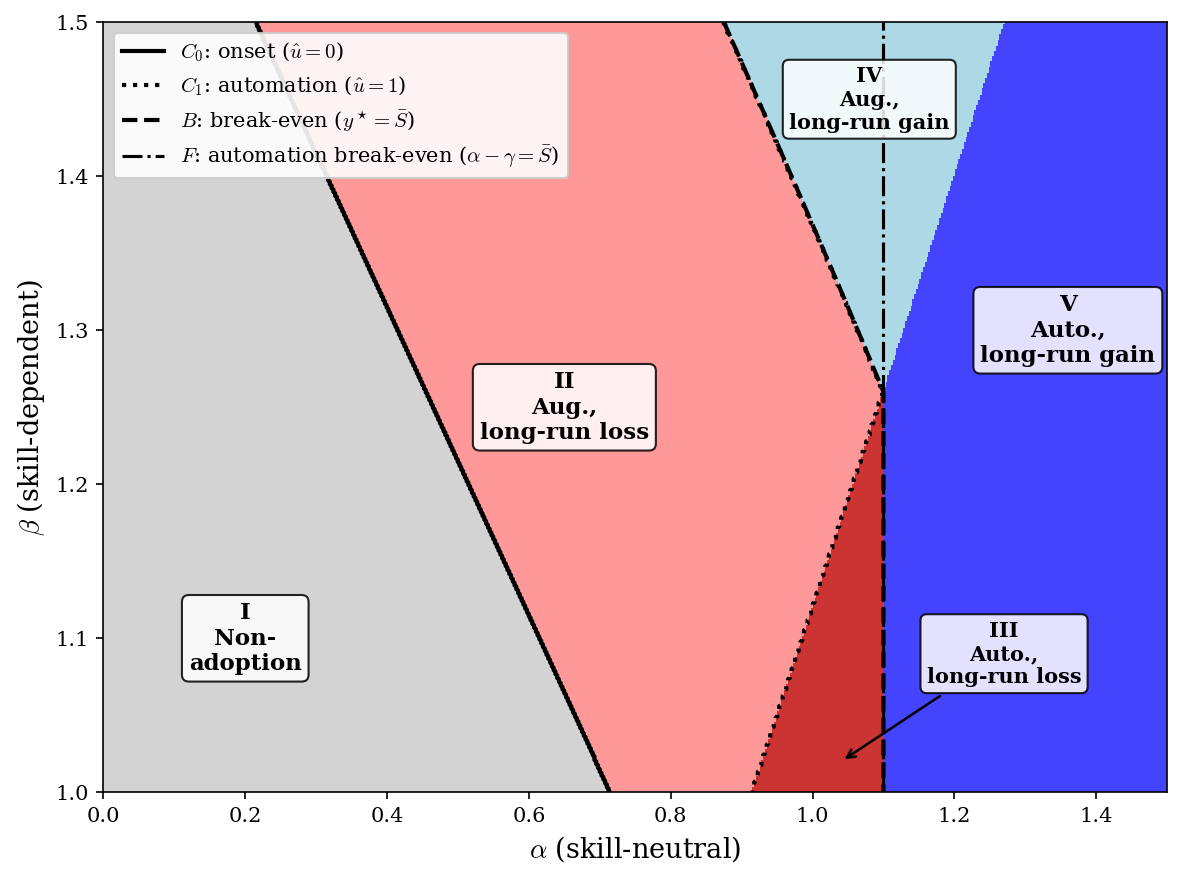}
\caption{\textbf{Five-region map.} The figure partitions the $(\alpha,\beta)$ parameter space by adoption, automation, and long-run output effects. Region II is the steady-state loss region: adoption is rational but steady-state output is below the no-AI benchmark. Parameters: $\gamma=0.10$, $\kappa=0.25$, $\delta=0.1$, $\bar S=1$.}
\label{fig:nomathheatmap}
\end{figure}

Moving right along the horizontal axis increases the skill-independent gain $\alpha$, and moving up the vertical axis increases the skill-dependent gain $\beta$. As these gains rise, usage approaches $\hat u=1$, full automation. The red regions mark the parameter combinations where using the tool lowers long-run productivity. The break-even boundary $B$, which separates the loss and gain regions under adoption, matters most for policy. Figure~\ref{fig:region-magnitude} in the Appendix shades the same map by the magnitude of the long-run effect.

\FloatBarrier

\subsection{Misaligned Incentives and the Augmentation Trap}\label{sec:misaligned-objectives}

So far, we have assumed that the person deciding to use AI is the same person doing the work. However, skill erosion becomes a problem when decision-makers mandating AI use are not the one paying the price in lost expertise. These conflicting interests can arise between shareholders and managers, between managers and employees, or even when a worker relies too heavily on AI to hit a short-term target, sacrificing their own long-term development. To keep things simple, we focus on the manager-worker dynamic. This covers scenarios where a company directly mandates AI, sets quotas so high that workers are forced to rely on it, or builds systems where delegating to AI is  the easiest option.

To formalize the idea, let the firm choose usage to maximize discounted value with discount rate $\delta_F$, while the worker's privately optimal usage corresponds to $\delta_W$, with $\delta_F>\delta_W$. The interior first-order condition gives
\[
u^*(S;\delta)=\frac{\alpha+(\beta-1) S-\kappa\bar S\,V_S(S;\delta)}{2\gamma},
\]
where $V_S(S;\delta)$ is the marginal continuation value of skill. Higher discounting reduces $V_S$, so
\[
V_S(S;\delta_F) < V_S(S;\delta_W)
\quad\Rightarrow\quad
u^*(S;\delta_F) > u^*(S;\delta_W)
\quad\text{(for interior solutions).}
\]
Since output is constant at a steady state, the value of remaining there is $p(\hat S,\hat u)/\delta$ whatever discount rate is used to evaluate it. Steady-state loss under a given policy therefore means that steady-state output falls below the no-AI level, $p(\hat S,\hat u)<\bar S$; the discount rate matters only through the policy it induces.

\begin{proposition}[Overuse under short-termism]
\label{prop:overuse}
Suppose $\delta_F>\delta_W$, and both the firm and the worker face interior optimal policies with stable steady states ($1+u_1\bar S>0$). Then $u^*(S;\delta_F)>u^*(S;\delta_W)$ for all $S$ in the interior region. The firm's steady-state skill is strictly lower, and the set of $(\alpha,\beta)$ pairs for which the worker experiences steady-state loss under the firm's policy is strictly larger than under the worker's own policy (Lemma~\ref{lem:loss-region}).
\end{proposition}

\begin{proof}[Proof sketch]
From the first-order condition $u^*(S;\delta) = [\alpha+(\beta-1) S - \kappa\bar S V_S(S;\delta)]/(2\gamma)$, usage is decreasing in $V_S(S;\delta)$.\footnote{We show in the Appendix that the value $V$ rises with skill and that a more patient decision-maker places a higher marginal value on skill.} A higher $\delta$ reduces the weight on future skill in the Bellman equation, which lowers the marginal value of skill for all $S$. Lemma~\ref{lem:loss-region} in the Appendix completes the argument by comparing the two loss regions.
\end{proof}

Managers may be short-termist, meaning they care less about long-term expertise and so push for more AI use. This type of manager adopts AI for an immediate productivity boost, while discounting the deskilling that follows. Figure~\ref{fig:trap-welfare} shows an illustrative example. 

\begin{figure}[htbp]
  \centering
  \includegraphics[width=0.8\linewidth]{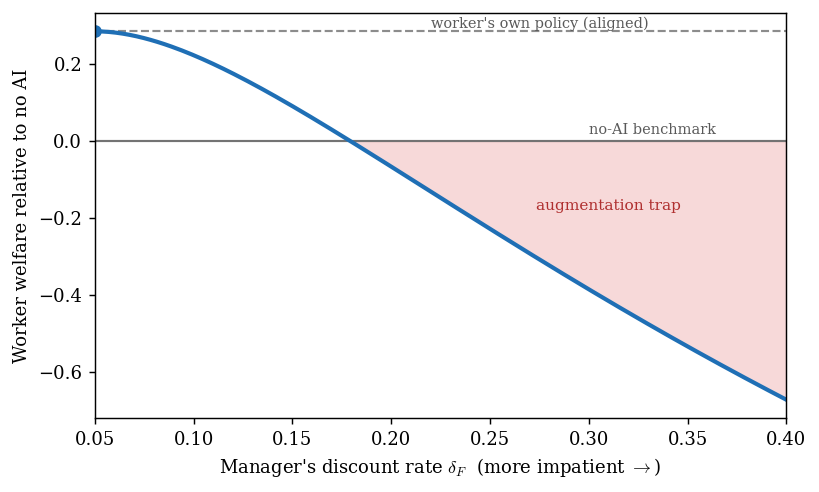}
  \caption{Misalignment and the augmentation trap. The figure plots worker lifetime welfare under the manager's policy, relative to no AI, as $\delta_F$ rises while $\delta_W=0.05$. The dashed line marks the aligned benchmark, $\delta_F=\delta_W$. Welfare declines with managerial impatience and crosses the no-AI benchmark at $\bar{\delta}_F\approx 0.18$; the shaded area is the trap region. The worked example uses $\delta_F=0.20$. Parameters: $\alpha=0.8$, $\beta=1.3$, $\gamma=1$, $\kappa=0.25$, $\bar S=1$.}
\label{fig:trap-welfare}
\end{figure}

\subsubsection{Worker Skill Externality}\label{sec:outside-options}

Beyond managerial short-termism, a second source of misalignment is that workers may value skill for reasons the firm's objective ignores: side projects, intellectual communities, independent understanding. These returns are an externality to the usage decision, because the decision-maker's objective does not account for them. We capture this with a term $\omega S_{it}$ ($\omega \ge 0$) added to the worker's payoff, so that per period the worker receives
\begin{equation}
y_{it}=(1-u_{it})S_{it}+\omega S_{it}+[\alpha+\beta S_{it}-\gamma u_{it}]u_{it}.
\end{equation}

\begin{proposition}[Worker skill externality reduces AI usage]
\label{prop:outside-options}
Fix $\alpha,\beta,\gamma,\kappa,\delta>0$, $\bar S>0$, and $\omega\ge 0$. The optimal usage policy with the worker skill externality is
\begin{equation}
u^*(\omega, S_{it})=u_0-u_{\omega}\,\omega+u_1 S_{it},
\end{equation}
where
\[
u_{\omega}=\frac{\kappa \bar S}{\gamma\delta+\sqrt{[\kappa((\beta-1)\bar S+2\gamma)+\gamma\delta]^2-\kappa^2(\beta-1)^2\bar S^2}}>0.
\]
At $\omega=0$ the worker's and firm's policies coincide. The steady-state skill level is
\begin{equation}
\hat S(\omega)=\frac{\bar S(1-u_0+u_{\omega}\omega)}{1+u_1\bar S},
\end{equation}
which is strictly increasing in $\omega$ when the interior steady state is stable ($1+u_1\bar S>0$), so that workers who value their skill more highly preserve more of it.
\end{proposition}

When the worker and firm are aligned on discount rate $\delta$ but the worker has positive skill value $\omega$, usage is over-prescribed because the firm ignores this private return. Thus, a worker with a large enough $\omega$ is worse off than with no AI (Proposition~\ref{prop:trap-conditions}).

We measure welfare as the present discounted value of the worker's per-period payoff, using the worker's own discount rate $\delta_W$. When the worker values skill beyond its contribution to output, that payoff includes the private return $\omega S$. In the parametric model,
\[
V_W \;=\; \int_0^\infty e^{-\delta_W t} \bigl[\, p(S(t),\, u(t)) \;+\; \omega\, S(t) \,\bigr]\, dt.
\]

\begin{definition}[Augmentation trap]
\label{def:trap}
Let $V_W^{\text{no-AI}}$ denote the worker's lifetime welfare absent AI, and let $V_W(u^*)$ denote the worker's lifetime welfare under the equilibrium usage policy $u^*$. A deployment is in the \emph{augmentation trap} if
\[
V_W(u^*) \;<\; V_W^{\text{no-AI}}.
\]
\end{definition}

\begin{proposition}[Conditions for the augmentation trap]
\label{prop:trap-conditions}

\begin{enumerate}
\leavevmode
\item When the decision-maker and the worker are aligned ($\delta_F = \delta_W$, $\omega = 0$), $V_W(u^*_W) \ge V_W^{\text{no-AI}}$. Mechanically, a worker choosing their own usage does not fall into the trap.

\item Under discount-rate divergence ($\delta_F > \delta_W$), the decision-maker's policy weakly raises usage relative to the worker's preferred policy, and strictly so at skill levels where both policies are interior. Suppose steady-state output remains below the no-AI level in the myopic limit $\delta_F \to \infty$. Then there is a threshold $\bar\delta_W > 0$ such that for every $\delta_W < \bar\delta_W$, $V_W(u^*_F) < V_W^{\text{no-AI}}$ for all sufficiently large $\delta_F$. Both thresholds depend jointly on all model parameters.
\item Under the worker skill externality ($\omega > 0$), the firm's policy ignores the private return to skill, and the trap appears for sufficiently large $\omega$.
\end{enumerate}
\end{proposition}

A worker that sets their own usage chooses their optimal policy, so even if using AI is always welfare-decreasing they can just choose not to use it. The trap requires misalignment between the usage decision-maker and the bearer of skill costs. In our expository example it appears when the worker's career is long enough that the permanent loss outweighs the temporary gain, or when the firm ignores what the worker's skill is worth outside the job. The Appendix gives sufficient conditions for parts (2) and (3).

The steady-state loss region in Figure~\ref{fig:nomathheatmap} identifies where the trap can occur. In the limit $\delta_W \to 0$, the worker weights long-run productivity infinitely more than the transition path, so their welfare ranking reduces to the steady-state comparison. In this case, Region~II depicts an augmentation trap. For $\delta_W > 0$, the trap region is a subset of Region~II, because the transition surplus offsets the steady-state deficit.

\subsection{Skill Stratification}\label{sec:stratification}

The preceding analysis focuses on interior policies, where \(0<u^*(S)<1\) and workers converge to a single steady state. A different dynamic arises when AI is sufficiently skill-substituting that the interior usage rule reaches the natural bounds \(u=0\) and \(u=1\). When \(\beta<1\), the interior policy slopes downward in skill: lower-skill workers use AI more heavily, so \(u_1<0\). If this slope is steep enough, the interior rule assigns low-skill workers to full automation and high-skill workers to no automation.

At the two boundaries, the skill dynamics are simple. Under full automation, \(u=1\), practice is fully displaced and skill converges to zero. Under no automation, \(u=0\), the worker practices unaided and skill converges to \(\bar S\). Stratification occurs because the threshold separating these regions is unstable, so workers below it move toward full automation and skill loss, while workers above it move toward no automation and skill recovery.

The interior rule reaches both usage boundaries when
\[
(1-\beta+2\kappa a\bar S)\bar S > 2\gamma,
\]
where \(a\) is the quadratic value-function coefficient from Lemma~\ref{lem:quad-value}. When this condition holds and
\[
\alpha \in \left(2\gamma+\kappa b\bar S,\; (1-\beta)\bar S+\kappa b\bar S+2\kappa a\bar S^2\right),
\]
the feasible usage rule contains both a full-automation region and a no-automation region.

\begin{proposition}[Permanent skill stratification]\label{prop:stratification}
Suppose \(\beta<1\) and the interior usage rule \(\tilde u(S)=u_0+u_1S\) satisfies \(u_1<0\), \(\tilde u(0)>1\), and \(\tilde u(\bar S)<0\). Under the feasible usage rule that assigns \(u=1\) when \(\tilde u(S)\geq 1\), assigns \(u=0\) when \(\tilde u(S)\leq 0\), and follows the interior rule otherwise, there exists an unstable threshold \(S_{\mathrm{eq}}\in(0,\bar S)\). Workers with initial skill \(S_0<S_{\mathrm{eq}}\) converge to \(\hat S=0\), while workers with \(S_0>S_{\mathrm{eq}}\) converge to \(\hat S=\bar S\).
\end{proposition}

The expression for \(S_{\mathrm{eq}}\) is the same as the interior steady state \(\hat S\) in \eqref{eq:Shat}. The difference is stability. In the interior stable case, \(1+u_1\bar S>0\), and nearby trajectories converge to \(\hat S\). In the stratification case, \(1+u_1\bar S<0\), and the same crossing repels nearby trajectories.

A lower \(\kappa\) means that skill is slower to rebuild once practice is displaced. This makes heavy offloading more attractive for an impatient decision-maker and raises the policy intercept \(u_0\), making the full-automation condition easier to satisfy. The mechanism is therefore most relevant in settings where expertise accumulates slowly. Figure~\ref{fig:k-curve} illustrates such a case with \(\kappa=0.1\), so the time scale for skill recovery is long.

For example, consider a software team that deploys the same coding assistant across junior and senior engineers. If the effective workflow is skill-substituting, so that lower-skill engineers receive larger marginal gains from offloading, the policy induces heavier reliance among junior engineers and lighter reliance among senior engineers. Under the stratification conditions, junior engineers move toward full automation and stop accumulating the underlying coding skill, while senior engineers continue practicing and converge toward their skill potential. The result is a widening skill gap within the team.

\begin{figure}[htbp] \centering \includegraphics[width=0.8\linewidth]{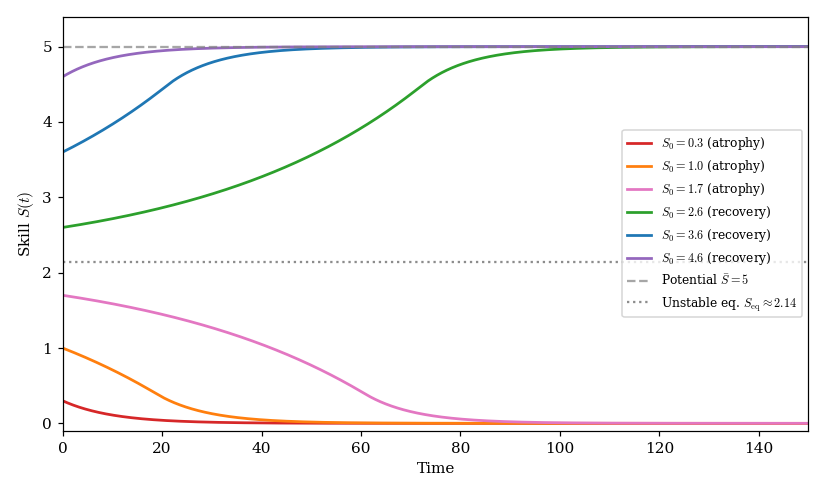} \caption{\textbf{Skill stratification.}
The downward-sloping usage rule creates an unstable threshold. Workers below the threshold rely heavily on AI and deskill; workers above it limit use of AI and converge to $\bar S$. Parameters: $\beta=0.71$, $\alpha=3.85$, $\bar S=5$, $\kappa=0.1$, $\gamma=1$, $\delta=0.1$.} \label{fig:k-curve} \end{figure}

\subsection{Misalignment and Skill Divergence}

\subsubsection{Skill Divergence}\label{sec:misalign-divergence}

When the stratification condition holds, a small difference in the decision-maker's objective can move a worker from the high-skill steady state to the low-skill one. The threshold $S_{\mathrm{eq}}$ depends on the policy intercept $u_0$, which is increasing in the discount rate $\delta$ and decreasing in the skill externality $\omega$. A manager with a slightly higher discount rate raises $u_0$, which shifts $S_{\mathrm{eq}}$ upward, potentially moving a worker from above the threshold to below it. Under the worker's own policy, this worker would have converged to $\bar S$. Under the firm's policy, they converge to zero.

\begin{proposition}[Misalignment and skill divergence]\label{prop:basin-switch}
Suppose $\beta<1$, that $(1-\beta + 2\kappa a\bar S)\bar S > 2\gamma$, and that the unconstrained linear policy satisfies the conditions of Proposition~\ref{prop:stratification} ($\tilde u(0)>1$ and $\tilde u(\bar S)<0$), so that the feasible policy produces the population split. The threshold $S_{\mathrm{eq}}$ is strictly increasing in $\delta$ and strictly decreasing in $\omega$. Consequently:
\begin{enumerate}\setlength{\itemsep}{2pt}
\item For any worker with initial skill $S_0 \in (S_{\mathrm{eq}}(\delta_W),\; S_{\mathrm{eq}}(\delta_F))$, the worker's own policy produces convergence to $\bar S$, while the firm's policy produces convergence to $0$.
\item For any worker with initial skill $S_0 \in (S_{\mathrm{eq}}(\omega),\; S_{\mathrm{eq}}(\omega=0))$, the worker's own policy (which accounts for $\omega$) produces convergence to $\bar S$, while the firm's policy (which ignores $\omega$) produces convergence to $0$.
\end{enumerate}
\end{proposition}

Under the positive feedback, a small difference in objectives produces a binary difference in outcomes. A manager with a shorter planning horizon can push a worker from skill accumulation into deskilling. As in Proposition~\ref{prop:stratification}, this result describes what happens when the optimal interior policy reaches the usage bounds $u=0$ and $u=1$. 

\FloatBarrier
\subsection{Transition Dynamics}

Where the interior policy is stable, skill under optimal AI usage converges to $\hat S$ at effective rate $\hat\kappa$, while without AI it would converge to $\bar S$ at rate $\kappa$.

Figure~\ref{fig:ct-trajectories} illustrates productivity, skill, and usage dynamics for representative points in Regions I--V. In Region~II, output jumps at adoption, exceeds the no-AI baseline for a time, and then falls below the pre-adoption level as skill erodes. Adoption remains privately rational because the transition surplus exceeds the long-run shortfall in present-value terms.

\begin{figure}[t!]
\centering
\includegraphics[width=0.65\linewidth]{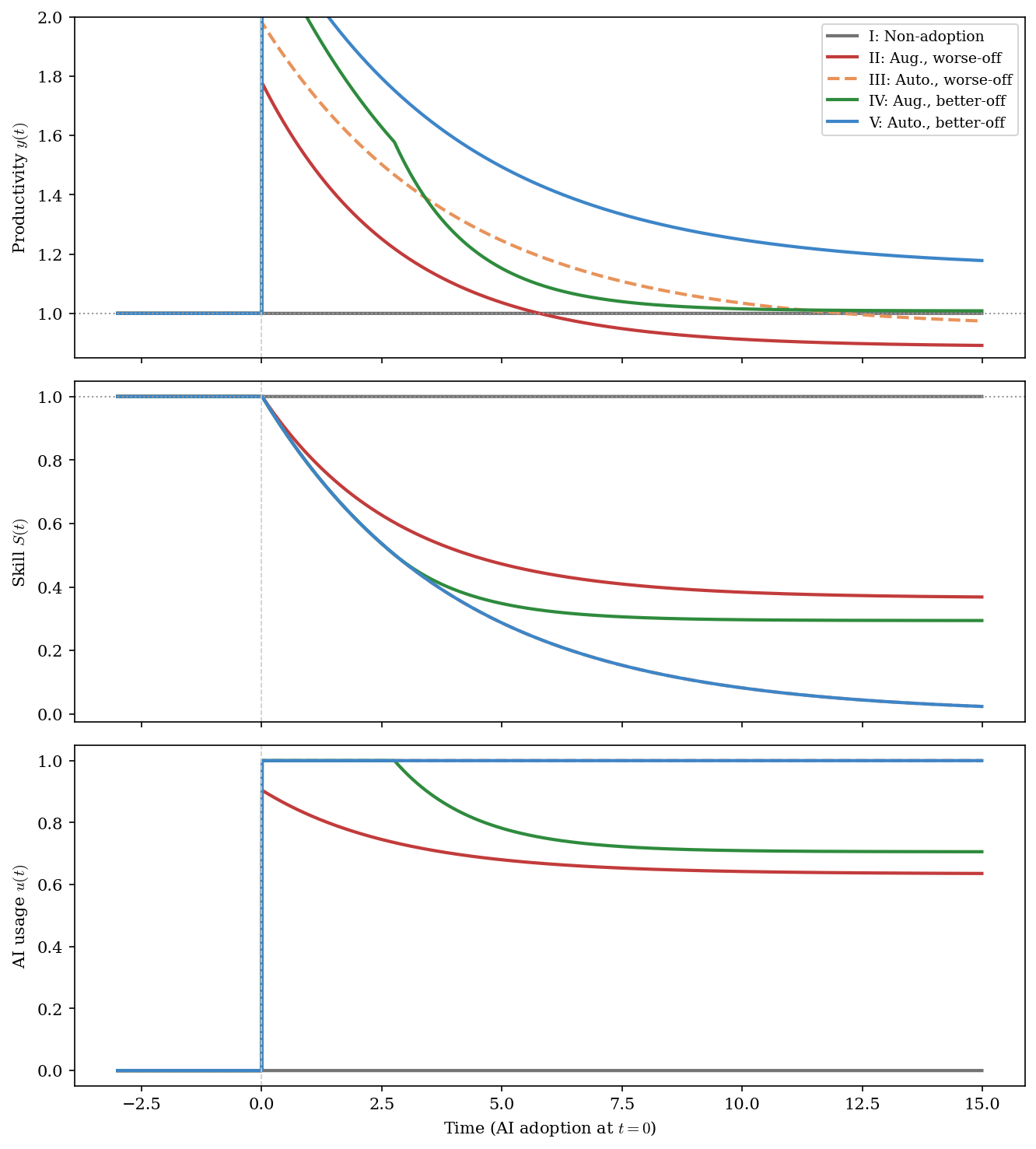}
\caption{\textbf{Transition dynamics across regions.}
The panels plot output, skill, and offloading after adoption for representative parameter points from Regions I--V in Figure~\ref{fig:nomathheatmap}. Region II illustrates augmentation with steady-state loss as skill erodes. Region III and Region V share the same skill and offloading paths under full automation. Parameters: $\gamma=0.10$, $\kappa=0.25$, $\delta=0.1$, $\bar S=1$.}
\label{fig:ct-trajectories}
\end{figure}
\FloatBarrier

Because the optimal policy is linear in skill, regime switches can occur as skill evolves. A worker may delay adoption until skill reaches a threshold, or move from full automation back to interior augmentation as skill decays. The switching thresholds and time-to-entry expressions are derived in the Appendix.

If a worker stops relying on AI and returns to doing the work themselves, skill recovers toward $\bar S$ at rate $\kappa$. 

\FloatBarrier

\section{Discussion}\label{sec:Discussion}

As more experiments document skill atrophy, it will be important to understand when it occurs and how its consequences should be managed. Organizations need measures that capture both the cost of skill loss and the importance of the affected skills. When the skill-productivity trade-off is present, misaligned incentives can also lead workers into an augmentation trap. This section translates the model into organizational implications: measuring the relevant skills and their complementarity with AI, redesigning workflows to preserve practice, using training to rebuild capability, aligning incentives, and distinguishing performance extraction from skill preservation. 

\subsection{Organizational Implications}\label{sec:diagnosis and levers}

Not all skill atrophy is harmful. When a technology makes a skill obsolete, preserving that skill may have little economic value: compilers reduced the need for many programmers to write assembly, and spreadsheets reduced the need for many forms of manual calculation. The risk we study arises when AI offloads practice in skills that remain valuable for AI-assisted production, such as problem formulation, error detection, contextual judgment, and knowing when not to trust the system. Because artificial intelligence tools can either automate obsolete work or erode complementary expertise, organizations need to identify which skills remain productive with AI and whether the workflow preserves the practice through which those skills are maintained.

\textbf{Measurement.} To understand the productivity effects of a deployment, firms need to know how much output depends on worker skill. This also requires measuring the quality and relevance of workers' skills over time. By relevance, we refer to the skills workers need to supervise AI well, since these will remain valuable as AI becomes pervasive. Section~\ref{sec:design} examines this issue in more detail. Measurement is complicated by the fact that output in the presence of AI is no longer a reliable proxy for skill. Strong performance may reflect the worker, the tool, or their interaction. A firm that cares about long-run capability therefore has to measure skill more directly, perhaps through periodic unaided assessments or evaluations of workers' reasoning. Without these measurements, a firm cannot tell whether skill loss is pushing a deployment toward the augmentation trap. Estimating $\beta$ is also important because it indicates the degree of complementarity or substitution between AI and expertise, which in turn shapes whether skill converges or diverges across the workforce.

\textbf{Workflow design.} Once these measures are established, if they indicate that a deployment is eroding relevant skill or moving toward the augmentation trap, the most direct response is usually to redesign the workflow. If AI is introduced so that the worker drops out of judgment and decision-making, short-run output may rise while skill deteriorates through disuse. If instead it is integrated so that the worker still has to evaluate, guide, and correct its output, the continued exercise of human judgment can help maintain expertise. Workflow redesign may therefore be one mechanism through which organizations move deployments from Region II toward Region IV as they discover the complements needed for durable AI productivity.

Firms can also shape the effective productivity parameters $\alpha$ and $\beta$ so that workers are encouraged to use AI in a skill-preserving way. This can be as simple as providing guidelines for AI use and encouraging workers to maintain cognitive vigilance. If AI productivity gains require cognitive offloading, then training and aligned incentives may be important. 

\textbf{Training.} Training can limit the impact of skill atrophy through two channels. The first increases the learning rate $\kappa$. Increasing $\kappa$ shrinks the steady-state loss region (Figure~\ref{fig:kappa-delta}). The second counteracts skill loss by rebuilding eroded skill through periods of practice in which the worker works without AI. This raises the worker's skill level and corresponds in the model to intervals of $u=0$ that move skill back toward maximum potential $\bar S$.  Initial skill matters too if skill stratification is a concern. As Section~\ref{sec:stratification} shows, starting conditions can determine whether a worker remains on a skill-building path or drifts into deskilling. In settings where skill divergence can occur, Proposition~\ref{prop:basin-switch} suggests that initial skill can define worker trajectories. This implies an initial period of skill development before intensive AI use may help workers enter the trajectory associated with continued skill development.

\begin{figure}[htbp]
  \centering
  \includegraphics[width=\linewidth]{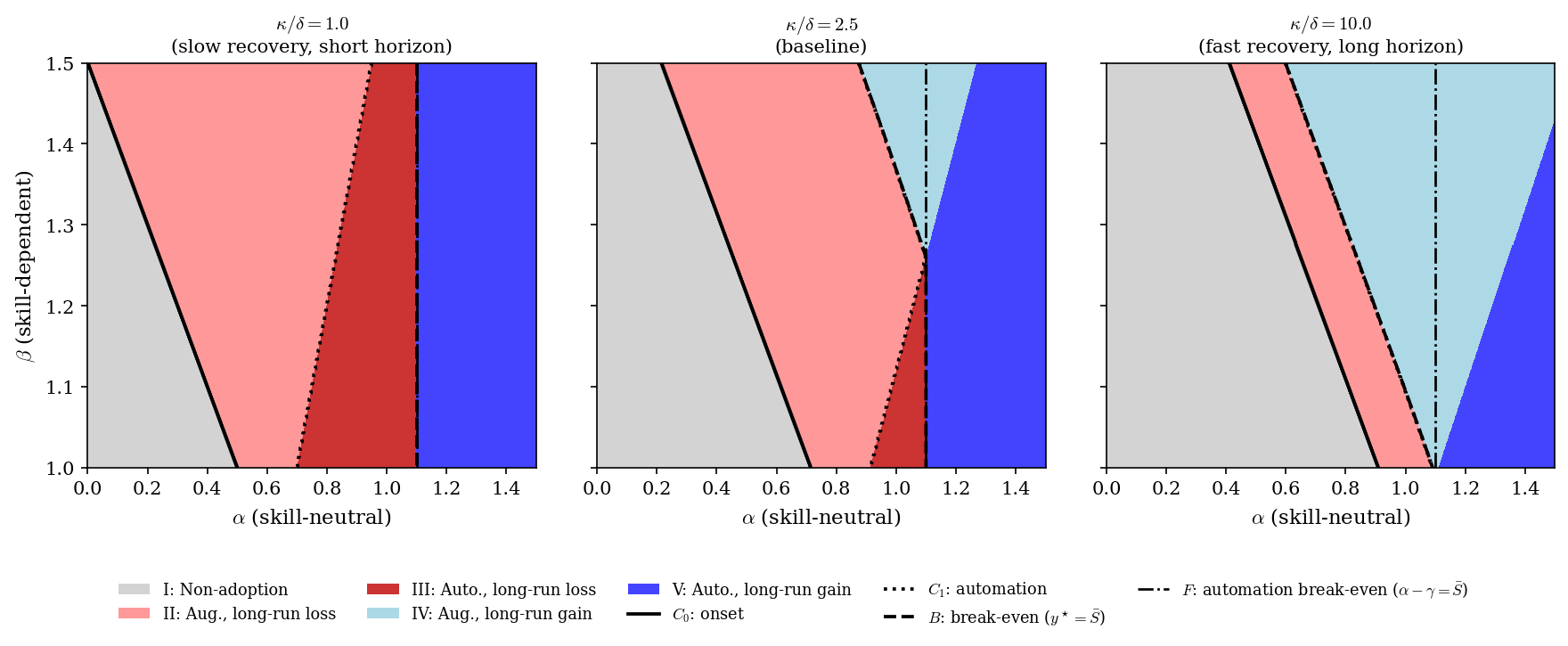}
  \caption{\textbf{Learning speed and the loss region.}
The ratio $\kappa/\delta$ captures how quickly skill recovers relative to the decision-maker's horizon. As $\kappa/\delta$ rises, the steady-state loss region contracts because the return to preserving skill is realized sooner, restraining offloading. Panels show $\kappa/\delta=1.0$, $2.5$, and $10.0$. All panels use $\gamma=0.10$ and $\bar S=1$.}
  \label{fig:kappa-delta}
\end{figure}

\FloatBarrier

\textbf{Incentive alignment.} The trap can arise because the decision-maker often does not bear the long-run cost of skill loss. A manager evaluated on current output discounts future worker capability more heavily than the worker, which leads to excessive AI use from the worker's perspective (Proposition~\ref{prop:overuse}). Extending the manager's evaluation horizon, or tying rewards to the development of the workers they supervise, could lead the manager to internalize more of those downstream consequences. We use the example of the short-termist manager for illustration, but misalignment can occur in any case where the decision-maker and the bearer of the skill costs differ. For example, the gap can also open between shareholders and managers, or within a single worker facing short-term pressure. As discussed previously, a source of misalignment can arise when workers value their skill beyond its value to the firm (Proposition~\ref{prop:outside-options}). 

\subsection{Performance Extraction and Skill Preservation}\label{sec:design}

We suggest that there are two different paradigms for organizational structures in the presence of deskilling AI. Some organizations may deploy AI in ways that draw on existing expertise while giving workers little opportunity to maintain or deepen it. Others may design AI use around the continued development of worker expertise. These two paradigms imply different paths for worker capability and for the long-run value of the deployment. The societal implications of this are interesting to consider given \citet{AcemogluAutorJohnson2026ProWorker}, who argue that markets undersupply pro-worker AI.

Call the first case \emph{performance extraction}. The worker uses their existing expertise in a workflow with heavy practice-displacing offloading. We consider this an extractive deployment because it draws on a stock of capability accumulated under the prior organization of work while doing little to maintain that stock. Such an organization could hire workers for expertise, assign them to workflows that offer little opportunity for skill development, and eventually replace them as their ability to judge AI output deteriorates. Although turnover is outside the model, this possibility points to a revolving-door employment model in which the organization's productivity comes from the consumption of expertise. In the case of misaligned incentives, such organizations represent an extreme case of the augmentation trap for workers. 

Call the second case \emph{skill preservation}. AI still changes the task and may eliminate substantial unaided work, but the workflow leaves the worker with the opportunity to develop the expertise that remains useful in AI-assisted production. \citet{ShenTamkin2026SkillFormation} offer some empirical support for this distinction, finding that delegating tasks to AI produced more skill loss than cognitively engaged use, although engaged use did not necessarily match the no-AI condition.

Just because a workflow formally requires human judgment does not mean that the worker is actually practicing that judgment. A nominally human-in-the-loop system remains extractive if the worker's interaction with the output is equivalent to rubber-stamping. Explainability can likewise create the appearance of engagement. Given the importance of individual differences, research from human-computer interaction will be useful for designing and personalizing workflows that sustain cognitive engagement. For example, \citet{JuAral2025Personality} find that how humans and AI agents are paired affects collaboration quality.

Short-run output is a poor way to distinguish performance extraction from skill preservation. A skill-preserving workflow may be slower, whereas an extractive workflow can raise measured output immediately by encouraging greater offloading and drawing on existing expertise. This type of extractive design will look superior because it is eroding the skill stock that sustains future performance. This measurement problem is emphasized by \citet{BondiJohnson2026}: when skill is endogenous to prior AI use, short-run estimates of assisted output can overstate the long-run value of a deployment. 

Regions~II and~IV in Figure~\ref{fig:nomathheatmap} should be read partly as organizational outcomes rather than inevitable consequences of using AI. A firm can change how any particular model is embedded in production, altering the effective values of $\alpha$ and $\beta$ and the resulting usage path. This interpretation also connects the model to the productivity J-curve for general-purpose technologies. In the standard account, firms realize the full gains from a new technology only after developing complementary investments \citep{BrynjolfssonRockSyverson2021}. In our setting, those complements can be understood as changing the effective $(\alpha,\beta)$ of a deployment. A firm may initially deploy an LLM as a generic Region II offloading tool. As the organization experiments with new ways of interacting with the tool, it may discover which tasks to delegate, when to require independent reasoning, and how to monitor errors. Codifying these discoveries can shift AI workflows into Region IV. 

In sum, an organizational risk is that AI allows firms to consume expertise faster than they replenish it. AI becomes a durable complement when workflow design preserves the practice through which valuable judgment is maintained; otherwise, AI can become an extractive technology that consumes worker expertise.

\subsection{Testable Predictions}\label{sec:predictions}

The model makes conditional predictions. When AI use displaces practice in tasks that remain important for supervising or improving output, heavier offloading should reduce subsequent unaided performance. This could be tested by examining worker behavior around temporary loss of access events or periods of unassisted work. 

The model also predicts that the relationship between worker skill and AI use depends on complementarity. On tasks for which expertise matters, more experienced workers should either use AI more intensively or gain more. When AI substitutes for expertise, less experienced workers should rely more heavily on the tool. An organization can infer the degree of complementarity or substitution of AI for skill by comparing usage and performance patterns across workers with different baseline expertise.

Because AI raises current output immediately while skill adjusts gradually, the model predicts a distinctive time path after adoption.  When skill atrophy occurs, the initial productivity gain may therefore give way to a deterioration in long-run capability. If workers later return to unaided practice, skill should recover over time rather than remain permanently fixed at its post-adoption level. 

The intensity of adoption should also depend on the user. Units managed for short-run throughput should choose heavier AI use than otherwise similar units evaluated on longer-run performance or worker development. Conversely, workers with larger private returns to expertise should prefer less intensive use. 

Finally, in deployments where AI substitutes strongly for expertise and the induced usage policy reaches the usage bounds, the model predicts skill divergence. The mechanism is threshold-based. Workers below the threshold $S_{eq}$ are pushed toward heavy offloading or full automation, so their unaided skill decays; workers above the threshold use the tool less and continue practicing, so skill recovers. Small differences in initial skill, or in decision-maker objectives that shift the threshold, can therefore produce large differences in long-run outcomes for workers near the threshold. 

A decline in unaided performance is evidence for skill atrophy, but such a situation may not indicate an augmentation trap. Some atrophy may be rational if the displaced skill has little future value, or if the early gains from AI compensate for lower unaided capability. Empirical work can try to understand what skills will remain useful to production with AI systems.  Measures such as revision depth, override rates, accepted suggestions, and error detection can help distinguish workflows that merely preserve assisted output from those that preserve the expertise needed for high-quality output.

\subsection{Limitations}\label{sec:limitations}

The model uses one skill dimension and fixed parameters so that the main parameter space can be solved in closed form. A central scope condition is that the deployment involves practice-displacing offloading: AI takes over activity that would otherwise build a skill that remains valuable for production with AI. An important simplifying assumption is the particular one-dimensional learning-forgetting process used to model this mechanism. This specification does not cover deployments in which AI actively builds the relevant skill. If AI builds productive skill rather than replacing the practice that builds it, the augmentation-trap outcome is mitigated. The one-dimensional skill representation also abstracts from variation in the quality of expertise. In the model, skill affects output only through its level, but in many settings the most important effect of expertise is on quality. These costs may appear first in quality-sensitive outcomes rather than in routine productivity measures.

The model also abstracts from skill recomposition. Because skill is represented by a single state variable, the model treats AI-induced capability change as movement along one dimension. In practice, AI may erode some skills while developing others. The effect of AI may be that delegation of routine coding allows for more practice of other skills such as system design. Future work could model skill as a vector, allowing AI to deplete some dimensions of expertise while building others that are useful for directing and supervising AI.

The parameters $(\alpha,\beta)$ are also treated as exogenous when in practice they change over time. The productivity gain parameters $\alpha$ and $\beta$ can change depending on model development and as users change how they interact with AI. For example, the post-adoption literature documents wide variation in how actively workers engage with systems, with many settling into narrow, routinized forms of use \citep{JaspersonEtAl2005,BurtonJonesStraub2006}. If engagement with AI output declines, effective $\beta$ falls even when the design of the tool is unchanged. Language models may also produce sycophantic responses rather than challenging users \citep{SharmaEtAl2024Sycophancy}. When the tool affirms weak reasoning, workers may become complacent, which may cause effective $\beta$ to fall further. A deployment can therefore move across regions over time. Routinized use and declining engagement may lower effective $\beta$ and push a deployment into a steady-state loss region, while complementary investments in workflow design, measurement, and training may raise effective $\beta$ and move a deployment from Region II toward Region IV.

The analysis also abstracts from task specialization and labor markets. If AI erodes a skill across many workers at once, the market value of that skill changes, which can alter the optimal policy. Training pipelines such as residency programs depend on experienced practitioners who can teach. If enough experienced workers deskill, there may eventually be too few left to train the next cohort. \citet{AcemogluKongOzdaglar2026KnowledgeCollapse} study a related problem, showing how AI can degrade a shared knowledge stock when individual firms do not internalize the externality. 

\section{Conclusion}\label{sec:conclusion}
A distinctive feature of current AI systems is the flexibility of the boundary between human effort and automated assistance. Many earlier technologies automated relatively well-defined tasks. AI can instead be introduced at many points in cognitive work, so the extent of offloading is shaped by workers, managers, and workflow design. In principle, this flexibility can create complementarity between human and artificial intelligence. In practice, the pressure to produce can lead organizations to offload the activities through which workers acquire and maintain the expertise needed to use AI well. Our model studies the consequences of where this boundary is set. Our analysis suggests that even with full awareness of skill atrophy, AI productivity may come with potentially hidden impacts such as reduced productivity, the deskilling of augmented workers, and divergence in the worker population. 

These results do not imply that skill atrophy is always harmful or that AI adoption should be restrained in general. Some skills will become obsolete, and preserving them would have little value. Our concerns about the augmentation trap arise when the intensity of tool usage excessively erodes valuable expertise. Current organizations do not yet know how to mitigate such costs.

Our framework identifies theoretical conditions under which these problems can arise, but practical diagnostics remain underdeveloped. Organizations typically measure AI productivity increases; few measure worker skill and its role in productivity with AI systems. Developing such measures should be a priority for researchers and firms. Workflow design, training, and incentive alignment may reduce these risks if they preserve relevant forms of expertise. Workers should be wary of how their AI workflows affect their skills, firms should consider measuring skill directly, and regulators should pay attention to the long-term effects of AI on expertise.

\paragraph*{Acknowledgments.}
We thank Dean Eckles, Glenn Ellison, Anna Stansbury, Thomas Ma, Haiwen Li, and Peyman Shahidi, as well as members of the Social Analytics Lab, for encouraging and insightful conversations.

\bibliography{refs}

\newpage
\appendix
\renewcommand{\thesection}{EC.\arabic{section}}
\renewcommand{\theequation}{EC.\arabic{equation}}
\renewcommand{\thefigure}{EC.\arabic{figure}}
\renewcommand{\thetable}{EC.\arabic{table}}
\renewcommand{\theproposition}{EC.\arabic{proposition}}
\renewcommand{\thelemma}{EC.\arabic{lemma}}
\renewcommand{\thecorollary}{EC.\arabic{corollary}}
\renewcommand{\theremark}{EC.\arabic{remark}}
\renewcommand{\thedefinition}{EC.\arabic{definition}}
\renewcommand{\theassumption}{EC.\arabic{assumption}}
\setcounter{equation}{0}
\setcounter{figure}{0}
\setcounter{table}{0}
\setcounter{proposition}{0}
\setcounter{lemma}{0}
\setcounter{section}{0}
\setcounter{remark}{0}
\setcounter{definition}{0}
\setcounter{assumption}{0}
\setcounter{corollary}{0}

\section{Dynamic Model: Solution and Proofs}\label{app:proofs}

This appendix contains all proofs for the dynamic model. The production function is $p = (1-u)S + [\alpha+\beta S - \gamma u]u$, which expands to $p = S + [\alpha + (\beta-1)S - \gamma u]u$. The term $\beta-1$ appears throughout the closed-form solutions because it captures the net effect of AI on the skill-dependent component: $\beta S$ gained through complementarity minus $S$ displaced from human contribution.
Section~\ref{app:extensions} presents extensions.

\subsection{Core Solution (Proof of Lemma~\ref{lem:quad-value})}\label{app:lemma}

\begin{proof}[Proof of Lemma~\ref{lem:quad-value}]
Expanding $p = (1-u)S + [\alpha+\beta S - \gamma u]u = S + [\alpha + (\beta-1)S - \gamma u]u$, the Hamilton--Jacobi--Bellman equation is
\begin{equation}
\label{equation:bellman}
    \delta V(S) =
    \max_{u}\; S + \alpha u - \gamma u^2 + (\beta-1) S\, u + V'(S)\,\kappa[\bar S(1-u) - S].
\end{equation}
The first-order condition for $u$ yields the optimal usage policy:
\begin{equation}\label{eq:foc}
  u^*(S) = \frac{\alpha + (\beta-1) S - \kappa \bar S\, V'(S)}{2\gamma}.
\end{equation}

Conjecture a quadratic value function $V(S) = aS^2 + bS + c$.
Then $u^*(S)$ is linear in $S$, and substituting into \eqref{equation:bellman} gives
\begin{eqnarray*}
    \delta (aS^2+bS+c) &=& S+\frac{[\alpha+(\beta-1) S-\kappa \bar S(2aS+b)]^2}{4\gamma}\\
    &&+\,(2aS+b)\,\kappa[\bar S-S].
\end{eqnarray*}

Matching the $S^2$, $S^1$, and $S^0$ coefficients:
\begin{align}
\delta a &= \frac{[(\beta-1)-2\kappa a\bar S]^2}{4\gamma}-2a\kappa,
\label{eq:coeff-a}\\[4pt]
\delta b &= 1+\frac{(\alpha-\kappa b\bar S)[(\beta-1)-2\kappa a\bar S]}{2\gamma}+2a\kappa \bar S-b\kappa,
\label{eq:coeff-b}\\[4pt]
\delta c &=
\frac{(\alpha-\kappa b\bar S)^2}{4\gamma}+b\kappa\bar S.
\label{eq:coeff-c}
\end{align}

Solving for the coefficients (selecting the stable root for $a$):
\begin{align}
a &= \frac{\kappa((\beta-1)\bar S+2\gamma)+\gamma\delta-\sqrt{D}}{2\kappa^2\bar S^2},
\label{eq:a-closed}\\[4pt]
b &= \frac{2\gamma(1+2a\kappa\bar S)+\alpha((\beta-1)-2\kappa a\bar S)}{2(\delta+\kappa)\gamma+\kappa(\beta-1)\bar S-2\kappa^2 a \bar S^2},
\label{eq:b-closed}\\[4pt]
c &=
\frac{(\alpha-\kappa b\bar S)^2}{4\gamma\delta}+\frac{b\kappa\bar S}{\delta},
\label{eq:c-closed}
\end{align}
where
\begin{equation}\label{eq:discriminant}
D = \bigl[\kappa((\beta-1)\bar S+2\gamma)+\gamma\delta\bigr]^2 - \kappa^2(\beta-1)^2\bar S^2
  = \gamma(\delta+2\kappa)\bigl(2\kappa(\beta-1)\bar S + \gamma(\delta+2\kappa)\bigr).
\end{equation}
\end{proof}

\begin{remark}[Linearity in $\alpha$]\label{rem:linearity}
Because $a$ in~\eqref{eq:a-closed} depends on $\beta-1$, $\gamma$, $\kappa$, $\delta$, and $\bar S$ but not on $\alpha$, the denominator of $b$ in~\eqref{eq:b-closed} is likewise $\alpha$-free; hence $b$ is linear in $\alpha$. It follows that $u_0$, $u_1$, the steady-state skill $\hat S$, and the steady-state usage $\hat u$ are all linear in $\alpha$ for each $\beta$. The adoption frontier $C_0$ ($\hat u=0$) and the automation frontier $C_1$ ($\hat u=1$) are exactly straight lines in $(\alpha,\beta)$ space. At $C_0$, substituting $u^*(\bar S)=0$ into \eqref{eq:coeff-a}--\eqref{eq:coeff-b} gives $2a\bar S+b = 1/(\delta+\kappa)$, the no-AI marginal value of skill, so $C_0$ is the line $\alpha+(\beta-1)\bar S = \kappa\bar S/(\delta+\kappa)$. At $C_1$, substituting $u_0=1$ gives $b=\beta/(\delta+\kappa)$, the marginal value of skill under full automation, so $C_1$ is the line $\alpha = 2\gamma+\beta\kappa\bar S/(\delta+\kappa)$. For each $\beta$, $\Delta\hat V$ is quadratic in $\alpha$, so the break-even boundary $B$ crosses each horizontal line at most twice.  In other words, $\alpha$ shifts usage uniformly across skill levels.  All nonlinear interaction between AI and skill enters through $\beta-1$ alone.
\end{remark}

\subsection{Proof of Proposition~\ref{prop:u-increasing}} \label{Proof_Prop_u_increasing}

\begin{proof}[Proof of Proposition~\ref{prop:u-increasing}]
From Lemma~\ref{lem:quad-value}, when $\beta>1$ and an interior policy is optimal, the value
function is quadratic,
\[
V(S) \;=\; aS^2 + bS + c,
\]
and the optimal usage policy is linear in skill,
\[
u^*(S) \;=\; u_0 + u_1 S
     \;=\; \frac{\alpha + ((\beta-1) - 2\kappa a\bar S)S - \kappa b\bar S}{2\gamma},
\]
so that
\[
\frac{du^*(S)}{dS} \;=\; u_1 \;=\; \frac{(\beta-1) - 2\kappa a\bar S}{2\gamma}.
\]
Thus it suffices to show that $(\beta-1) - 2\kappa a\bar S > 0$ for the stable interior
solution.

From \eqref{eq:a-closed}, the quadratic coefficient $a$ is
\[
a \;=\;
\frac{\kappa((\beta-1)\bar S + 2\gamma) + \gamma\delta
      - \sqrt{D}}{2\kappa^2\bar S^2},
\]
with discriminant $D$ as in~\eqref{eq:discriminant}.
Multiplying by $2\kappa\bar S$ gives
\[
2\kappa a\bar S \;=\;
\frac{\kappa((\beta-1)\bar S + 2\gamma) + \gamma\delta - \sqrt{D}}{\kappa\bar S},
\]
so
\[
(\beta-1) - 2\kappa a\bar S
= \frac{(\beta-1)\kappa\bar S
        - \kappa((\beta-1)\bar S + 2\gamma)
        - \gamma\delta
        + \sqrt{D}}{\kappa\bar S}
= \frac{\sqrt{D} - \gamma(\delta + 2\kappa)}{\kappa\bar S}.
\]

Next, factor $D$:
\begin{align*}
D
&= \bigl[\kappa((\beta-1)\bar S + 2\gamma) + \gamma\delta\bigr]^2
   - \kappa^2(\beta-1)^2\bar S^2 \\
&= \gamma(\delta + 2\kappa)\bigl(2\kappa(\beta-1)\bar S
                                + \gamma(\delta + 2\kappa)\bigr).
\end{align*}
Since $\gamma>0$, $\kappa>0$, $\bar S>0$, and $\beta-1>0$, we have
\[
D \;>\; \gamma^2(\delta + 2\kappa)^2
\quad\Rightarrow\quad
\sqrt{D} \;>\; \gamma(\delta + 2\kappa).
\]
Therefore
\[
(\beta-1) - 2\kappa a\bar S
= \frac{\sqrt{D} - \gamma(\delta + 2\kappa)}{\kappa\bar S} \;>\; 0,
\]
and hence
\[
u_1 \;=\; \frac{(\beta-1) - 2\kappa a\bar S}{2\gamma} \;>\; 0.
\]

Thus the interior optimal policy $u^*(S)$ is strictly increasing in skill $S$
when $\beta>1$.
\end{proof}

\begin{proof}[Proof of Proposition~\ref{prop:u-decreasing}]
The same argument applies with the inequality reversed. When $\beta<1$, the interior quadratic solution exists provided $D > 0$, which requires $2\kappa(1-\beta)\bar S < \gamma(\delta+2\kappa)$. (When this condition fails, the policy hits the boundary and the stratification analysis of Section~\ref{app:kcurve} applies instead.) Given $D > 0$, the bracketed term $2\kappa(\beta-1)\bar S + \gamma(\delta+2\kappa)$ in the factorization of $D$ is strictly less than $\gamma(\delta+2\kappa)$, so $D < \gamma^2(\delta+2\kappa)^2$ and $\sqrt{D} < \gamma(\delta+2\kappa)$. Therefore $(\beta-1) - 2\kappa a\bar S = [\sqrt{D} - \gamma(\delta+2\kappa)]/(\kappa\bar S) < 0$, giving $u_1 < 0$.
\end{proof}

\begin{remark}[Usage is decreasing in potential when $\beta>1$]\label{rem:usage-potential}
Fix current skill $S\ge 0$ and suppose $\beta>1$ with an interior optimal policy. Then $\partial u^*(S;\bar S)/\partial\bar S < 0$: holding current skill fixed, a worker with higher potential uses less AI.
\end{remark}

\begin{proof}
From the first-order condition, $u^*(S;\bar S) = [\alpha+(\beta-1)S - \kappa (b\bar S) - 2\kappa (a\bar S) S]/(2\gamma)$, so it suffices to show that $a\bar S$ and $b\bar S$ are both strictly increasing in $\bar S$. Write $c_1 := \gamma(\delta+2\kappa)$ and $R := \sqrt{D}$, so that $R^2 - c_1^2 = 2\kappa(\beta-1)\bar S c_1$ by \eqref{eq:discriminant}, and note $dR/d\bar S = \kappa(\beta-1)c_1/R > 0$.

Since $\kappa((\beta-1)\bar S+2\gamma)+\gamma\delta = \kappa(\beta-1)\bar S + c_1 = (R^2-c_1^2)/(2c_1) + c_1$, the closed form \eqref{eq:a-closed} simplifies to
\[
a\bar S \;=\; \frac{(R-c_1)^2}{4\kappa^2 c_1 \bar S},
\qquad\text{and direct differentiation gives}\qquad
\frac{d(a\bar S)}{d\bar S} \;=\; \frac{(R-c_1)^2}{4\kappa^2\bar S^2 R} \;>\; 0.
\]

For $b$, the identities $\kappa\bar S\bigl(\beta-1-2\kappa a\bar S\bigr) = R-c_1$ and $2\gamma(\delta+\kappa)-c_1 = \gamma\delta$ let the closed form \eqref{eq:b-closed} be rewritten as
\[
b\bar S \;=\; \frac{2\gamma\bar S + \gamma(R-c_1)^2/(\kappa c_1) + \alpha(R-c_1)/\kappa}{\gamma\delta+R}.
\]
Differentiating with respect to $\bar S$, the numerator of $d(b\bar S)/d\bar S$ equals
\[
\frac{\gamma\bigl(R^2+2\gamma\delta R+c_1^2\bigr)}{R}
\;+\; \frac{\gamma(\beta-1)(R-c_1)\bigl(R+c_1+2\gamma\delta\bigr)}{R}
\;+\; \frac{\alpha(\beta-1)c_1(\gamma\delta+c_1)}{R},
\]
and each term is strictly positive for $\beta>1$ (where $R>c_1$). Hence $d(b\bar S)/d\bar S>0$, and
\[
\frac{\partial u^*}{\partial\bar S} \;=\; -\frac{\kappa}{2\gamma}\left[\frac{d(b\bar S)}{d\bar S} + 2S\,\frac{d(a\bar S)}{d\bar S}\right] \;<\; 0.
\]
\end{proof}

\subsection{Proof of Proposition~\ref{prop:overuse}}\label{app:overuse}

\begin{proof}
We show that $\partial V_S(S;\delta)/\partial\delta < 0$ for all $S\in[0,\hat S]$,
which by the first-order condition implies $u^*(S;\delta)$ is strictly increasing in~$\delta$.

At the optimum the HJB reads
\[
  \delta\, V(S) \;=\; p\bigl(S, u^*(S)\bigr) + V'(S)\,\dot S(S).
\]
Differentiate both sides with respect to $\delta$.  By the envelope theorem
the terms involving $\partial u^*/\partial\delta$ vanish, giving
\begin{equation}\label{eq:envelope}
  V(S) + \delta\, V_\delta(S) \;=\; V_{S\delta}(S)\,\dot S(S),
\end{equation}
where $V_\delta = \partial V/\partial\delta$ and $V_{S\delta} = \partial^2 V/\partial S\,\partial\delta$.

Since $V(S) = aS^2 + bS + c$ and $\dot S = \hat\kappa(\hat S - S)$ with
$\hat\kappa = \kappa(1 + u_1\bar S)$, both sides of~\eqref{eq:envelope} are
quadratic in~$S$.  Writing $V_\delta(S) = a_\delta S^2 + b_\delta S + c_\delta$
and matching the $S^2$ coefficient:
\[
  a + \delta\, a_\delta \;=\; -2\hat\kappa\, a_\delta,
  \qquad\text{so}\qquad
  a_\delta \;=\; \frac{-a}{\delta + 2\hat\kappa}.
\]

From $u_1 = \bigl(\sqrt{D} - \gamma(\delta+2\kappa)\bigr)/(2\gamma\kappa\bar S)$
where $D = \gamma(\delta+2\kappa)\bigl(2\kappa(\beta-1)\bar S + \gamma(\delta+2\kappa)\bigr)$,
a direct calculation gives
\[
  \delta + 2\hat\kappa
  = \delta + 2\kappa + \frac{\sqrt{D} - \gamma(\delta+2\kappa)}{\gamma}
  = \frac{\sqrt{D}}{\gamma},
\]
so $a_\delta = -\gamma a/\sqrt{D}$.
Since $D = A^2 - \kappa^2(\beta-1)^2\bar S^2 < A^2$ for $\beta\neq 1$ (where $A:=\kappa((\beta-1)\bar S+2\gamma)+\gamma\delta>0$), the stable root satisfies $a>0$ for all $\beta\neq 1$, giving $a_\delta < 0$.

Matching the $S^1$ coefficient in~\eqref{eq:envelope} gives
\[
  b_\delta = \frac{2a_\delta\,\hat\kappa\,\hat S - b}{\delta + \hat\kappa}.
\]

We need $w(S) \equiv V_{S\delta}(S) = 2a_\delta S + b_\delta < 0$ for all $S \in [0, \hat S]$.  Since $a_\delta < 0$,
the function $w$ is decreasing in~$S$, so it suffices to check $w(0) = b_\delta < 0$.

It remains to find the sign of $b$. A higher-skill worker can always follow the policy of a lower-skill worker; under the law of motion the skill gap then decays exponentially to zero while the higher-skill worker stays ahead. Because output is increasing in skill, the higher-skill worker's value under this suboptimal policy exceeds the lower-skill worker's, and the value under the optimal policy is higher. Hence $V$ is increasing in $S$, and differentiating at $S=0$ gives $b>0$.

\medskip\noindent
Since $a_\delta < 0$ and $b > 0$, the numerator of $b_\delta$ is $2a_\delta\hat\kappa\hat S - b < 0$, so $w(0) = b_\delta < 0$. Because $w$ is decreasing, $w(S) \le w(0) < 0$ for all $S \in [0,\hat S]$. That is,
$\partial V_S(S;\delta)/\partial\delta < 0$ for all~$S$,
so $u^*(S;\delta_F) > u^*(S;\delta_W)$ when $\delta_F > \delta_W$. The argument is identical for $\beta > 1$ and $\beta < 1$, since $a > 0$ in both cases.

Hence $u^*(S;\delta)$ is increasing in $\delta$. Higher usage under $\delta_F > \delta_W$ implies lower steady-state skill $\hat S(\delta_F) < \hat S(\delta_W)$. The next lemma compares the two loss regions.
\end{proof}

\begin{lemma}[The loss region grows with the discount rate]\label{lem:loss-region}
Fix $\gamma,\kappa,\bar S>0$ and $\delta_F>\delta_W>0$.
\begin{enumerate}
  \item At any $(\alpha,\beta)$ where both policies are interior with stable steady states, if the worker's policy produces steady-state loss, $p(\hat S(\delta_W),\hat u(\delta_W))<\bar S$, then so does the firm's.
  \item There is a nonempty set of $(\alpha,\beta)$ at which the firm's policy produces steady-state loss and the worker's produces none.
\end{enumerate}
Hence the firm's steady-state loss region strictly contains the worker's.
\end{lemma}

\begin{proof}
Write $\hat u(\delta)$ for steady-state usage and $y^\star(\hat u)$ for steady-state output as in \eqref{eq:y-steady}. Because output is constant at a steady state, its value there is $y^\star/\delta$ at any discount rate, so steady-state loss is the flow comparison $y^\star<\bar S$ and depends on $\delta$ only through the policy.

\emph{Step 1: $\hat u$ is increasing in $\delta$.} Steady-state usage solves $g(u):=u^*(\bar S(1-u);\delta)-u=0$, and $g'(u)=-(1+u_1\bar S)<0$ by stability. Since $u^*(S;\delta_F)>u^*(S;\delta_W)$ at every $S$ in the interior region, $g_F(\hat u(\delta_W))>0$, and because $g_F$ is decreasing, $\hat u(\delta_F)>\hat u(\delta_W)$.

\emph{Step 2: output relative to the benchmark.} Substituting $\hat S=\bar S(1-\hat u)$ into \eqref{eq:y-steady}, steady-state output differs from the no-AI benchmark by the productivity effect of AI usage, net of the human contribution that usage displaces and the output lost to atrophy:
\[
y^\star(\hat u)-\bar S
=\hat u\bigg[\;\underbrace{\alpha+\beta\hat S-\gamma\hat u}_{\substack{\text{AI productivity effect}\\ \text{per unit of usage}}}
\;-\;\underbrace{\hat S}_{\substack{\text{human contribution}\\ \text{displaced}}}
\;-\;\underbrace{\bar S}_{\substack{\text{output lost}\\ \text{to atrophy}}}\;\bigg].
\]
Collecting terms in $\hat u$, write this as $y^\star(\hat u)-\bar S=\hat u\, g(\hat u)$, where
\[
g(\hat u):=c_1+c_2\hat u,
\qquad
c_1=\alpha+(\beta-1)\bar S-\bar S,
\qquad
c_2=-\big(\gamma+(\beta-1)\bar S\big)
\]
is the net gain per unit of usage. The intercept $c_1$ is the local loss condition of Section~\ref{app:robustness}, the tool's net gain for a full-skill worker against the output a unit of usage displaces. The slope $c_2$ collects the losses that grow with usage: diminishing returns, and the complementarity destroyed as skill erodes. Under an adopting policy ($\hat u>0$), the steady state falls below the benchmark exactly when the net gain per unit of usage is negative.

\emph{Step 3: the case $c_2\le 0$.} The per-unit net gain $g$ is nonincreasing in usage. If it is negative at $\hat u(\delta_W)$, it is negative at the larger $\hat u(\delta_F)$, so the firm's policy also produces loss.

\emph{Step 4: the case $c_2>0$.} Here every interior policy produces loss, so the firm's policy produces loss whenever the worker's does. For any $\hat u<1$ the bracketed term is bounded above by $c_1+c_2=\alpha-\gamma-\bar S$, so a long-run gain would require $\alpha>\gamma+\bar S$. But $u_0$ is strictly increasing in $\alpha$, with $\partial u_0/\partial\alpha=(\delta+\kappa)/(\gamma\delta+\sqrt D)>0$ from \eqref{eq:b-closed}, and $u_0=1$ exactly on the automation line $\alpha=2\gamma+\beta\kappa\bar S/(\delta+\kappa)$ of Remark~\ref{rem:linearity}. An interior policy therefore requires $\alpha<2\gamma+\beta\kappa\bar S/(\delta+\kappa)$, and when $c_2>0$, that is when $\gamma<(1-\beta)\bar S$,
\[
2\gamma+\frac{\beta\kappa\bar S}{\delta+\kappa}-(\gamma+\bar S)
=\gamma-\bar S+\frac{\beta\kappa\bar S}{\delta+\kappa}
<-\beta\bar S+\frac{\beta\kappa\bar S}{\delta+\kappa}
=-\frac{\beta\bar S\delta}{\delta+\kappa}\le 0.
\]
So $\alpha<\gamma+\bar S$, and the bracketed term is negative at any interior usage level.

\emph{Step 5: Strictness.} Let
\[
A(\delta) \equiv \frac{\kappa \bar S}{\delta+\kappa}
\]
denote the threshold on the adoption frontier \(C_0\). Since \(A'(\delta)<0\), the assumption
\(\delta_F>\delta_W\) implies
\[
A(\delta_F)<A(\delta_W).
\]
Therefore there exist parameter values satisfying
\[
A(\delta_F)
<
\alpha+(\beta-1)\bar S
<
A(\delta_W).
\]
For such values, the firm's policy has positive steady-state usage, \(\hat u_F>0\), while the worker's
policy has zero steady-state usage, \(\hat u_W=0\). Moreover, by choosing
\(\alpha+(\beta-1)\bar S\) arbitrarily close to \(A(\delta_F)\) from above, \(\hat u_F\) can be made
arbitrarily small.

For these same parameter values,
\[
c_1
=
\alpha+(\beta-1)\bar S-\bar S
<
A(\delta_W)-\bar S
<0,
\]
because \(A(\delta_W)<\bar S\). Let
\[
g(u) \equiv c_1+c_2u .
\]
Since \(g(0)=c_1<0\), continuity implies that \(g(u)<0\) for all sufficiently small \(u>0\).
Choosing \(\alpha+(\beta-1)\bar S\) close enough to \(A(\delta_F)\) gives \(g(\hat u_F)<0\). Hence
\[
y^\star(\hat u_F)-\bar S
=
\hat u_F g(\hat u_F)
<0,
\]
while under the worker's policy,
\[
y^\star(\hat u_W)-\bar S
=
y^\star(0)-\bar S
=
0.
\]
Thus there exist parameter values for which the firm's policy produces steady-state loss while the
worker's own policy does not. The firm's loss region therefore strictly contains the worker's loss
region.

The conclusion also survives if the firm's policy reaches full automation. Worker-side loss requires $\alpha<\gamma+\bar S$: when $c_2\le 0$, interior loss needs $\hat u>c_1/(-c_2)$, and $\hat u<1$ then forces $c_1<-c_2$; when $c_2>0$, Step~4 applies. Full automation then yields steady-state output $\alpha-\gamma<\bar S$, again a loss.
\end{proof}

\subsection{Proof of Proposition~\ref{prop:outside-options}}\label{app:outside-options}

\begin{proof}
With the worker skill externality flow value $\omega S$, the worker's Bellman equation is
\[
\delta V(S) = \max_{u}\left\{(1+\omega)S + (\alpha+(\beta-1) S - \gamma u)u + V'(S)[\kappa\bar S(1-u) - \kappa S]\right\}.
\]
The first-order condition for $u$ is unchanged from the baseline:
\[
u^*(S) = \frac{\alpha + (\beta-1) S - \kappa\bar S V'(S)}{2\gamma}.
\]
Consider $V(S) = a(\omega)S^2 + b(\omega)S + c(\omega)$. Substituting and matching the $S^2$ coefficient gives the same equation as in the baseline, so $a(\omega) = a$ (independent of $\omega$).

Matching the $S^1$ coefficient yields an equation that is linear in $b$ with $\omega$ entering additively through the $(1+\omega)$ term:
\[
\delta b = (1+\omega) + ((\beta-1) - 2\kappa a\bar S)\frac{\alpha - \kappa\bar S b}{2\gamma} + 2a\kappa\bar S - b\kappa.
\]
Collecting all terms involving $b$ on the left-hand side and solving, $b(\omega)$ is linear in $\omega$ with
\[
\frac{\partial b}{\partial\omega} = \frac{2\gamma}{\gamma\delta + \sqrt{D}},
\]
where $D = [\kappa((\beta-1)\bar S + 2\gamma)+\gamma\delta]^2 - \kappa^2(\beta-1)^2\bar S^2$ as before.\footnote{To verify: set $\beta=1$, so $a=0$, $D = \gamma^2(\delta+2\kappa)^2$, $\sqrt{D} = \gamma(\delta+2\kappa)$, and $\partial b/\partial\omega = 2\gamma/(2\gamma\delta + 2\gamma\kappa) = 1/(\delta+\kappa)$. This matches $b = (1+\omega)/(\delta+\kappa)$ from the $\beta=1$ solution.}

Since $u_0 = (\alpha - \kappa\bar S b)/(2\gamma)$, we have $\partial u_0/\partial\omega = -(\kappa\bar S)/(2\gamma) \cdot \partial b/\partial\omega$, giving
\[
u_\omega := -\frac{\partial u_0}{\partial\omega} = \frac{\kappa\bar S}{\gamma\delta + \sqrt{D}} > 0.
\]

The steady-state skill $\hat S(\omega) = \bar S(1-u_0+u_\omega\omega)/(1+u_1\bar S)$ is strictly increasing in $\omega$ provided $1+u_1\bar S > 0$ (the stable interior condition), because $u_\omega > 0$.
\end{proof}

\subsection{Proof of Proposition~\ref{prop:trap-conditions}}\label{app:trap-conditions}

\begin{proof}

\textit{Aligned objectives ($\delta_F=\delta_W$, $\omega=0$).} The worker is the decision-maker.  The constant policy $u\equiv 0$ is feasible and yields
$V_W^{\text{no-AI}} = \bar S/\delta_W$.  Since $u^*_W$ maximizes $V_W$, we have
$V_W(u^*_W)\ge V_W^{\text{no-AI}}$, with equality if and only if $u^*_W=0$
(i.e., AI is not adopted).

\textit{Discount-rate divergence ($\delta_F>\delta_W$, $\omega=0$).} From Proposition~\ref{prop:overuse}, $u^*(S;\delta)$ is strictly increasing in~$\delta$ for every~$S$.
In particular, $u_0(\delta_F) > u_0(\delta_W)$ and $u_1(\delta_F) > u_1(\delta_W)$, because $\partial u_0/\partial\delta = -\kappa\bar S\,b_\delta/(2\gamma) > 0$ (since $b_\delta < 0$) and $\partial u_1/\partial\delta = -\kappa\bar S\,a_\delta/\gamma > 0$ (since $a_\delta < 0$).

The steady-state skill is $\hat S(\delta) = \bar S(1-u_0(\delta))/(1+u_1(\delta)\bar S)$.
Since both $u_0$ and $u_1$ are increasing in~$\delta$ and $\hat S$ is decreasing in~$u_0$
(for $1+u_1\bar S>0$, i.e.\ the stable interior case), $\hat S(\delta_F)<\hat S(\delta_W)$.
In the steady-state loss region, $p(\hat S,\hat u)<\bar S$ for all interior policies.

Starting from $S_0=\bar S$, skill evolves as $S(t) = \hat S_F + (\bar S-\hat S_F)\,e^{-\hat\kappa_F t}$
under the firm's linear policy, where $\hat\kappa_F = \kappa(1+u_{1,F}\bar S)>0$.
Because output $p(S,u^*_F(S))$ is quadratic in~$S$ and $S(t)$ is exponential in~$t$,
the worker's lifetime welfare can be written as
\[
  W(\delta_F)
  \;=\; \frac{p(\hat S_F,\hat u_F)}{\delta_W}
       + \frac{A}{\delta_W+\hat\kappa_F}
       + \frac{B}{\delta_W+2\hat\kappa_F},
\]
where $A$ and $B$ are bounded functions of the parameters that capture the
transient output deviation from steady state.
The no-AI benchmark is $V_W^{\text{no-AI}} = \bar S/\delta_W$.
The trap holds when
\[
  \frac{A}{\delta_W+\hat\kappa_F} + \frac{B}{\delta_W+2\hat\kappa_F}
  \;<\; \frac{\bar S - p(\hat S_F,\hat u_F)}{\delta_W}.
\]
As $\delta_F\to\infty$ the coefficients $a$ and $b$ vanish, so the firm's policy converges to the myopic policy $u^{\text{myo}}(S)=[\alpha+(\beta-1)S]/(2\gamma)$, and $W(\delta_F)$ converges to the worker's welfare under that policy, $W^{\text{myo}} = p(\hat S^{\text{myo}},\hat u^{\text{myo}})/\delta_W + T(\delta_W)$, where the transient term $T(\delta_W)$ remains bounded as $\delta_W\to 0$. By hypothesis the myopic steady state lies in the loss region, $p(\hat S^{\text{myo}},\hat u^{\text{myo}})<\bar S$, so the permanent deficit $[\bar S - p(\hat S^{\text{myo}},\hat u^{\text{myo}})]/\delta_W$ grows without bound as $\delta_W\to 0$ while the transient terms remain bounded. Hence there is a threshold $\bar\delta_W>0$ with $W^{\text{myo}}<\bar S/\delta_W$ for all $\delta_W<\bar\delta_W$. For such $\delta_W$, continuity of $W$ in $\delta_F$ gives $V_W(u^*_F)<V_W^{\text{no-AI}}$ for all sufficiently large $\delta_F$.\footnote{The condition on $\delta_W$ is not vacuous. A sufficiently impatient worker weights the transition surplus heavily enough that no discount-rate gap produces the trap. Monotonicity of $W(\delta_F)$ would establish a single threshold $\bar\delta_F$ above which the trap binds. We conjecture this holds generically but do not prove it.}

\textit{Worker skill externality ($\omega>0$, $\delta_F=\delta_W$).} The firm ignores $\omega$ and solves as if $\omega=0$.  The worker's welfare includes
the externality:
\[
  V_W^{\text{no-AI}} = \frac{(1+\omega)\bar S}{\delta_W},
  \qquad
  W_\omega = \int_0^\infty e^{-\delta_W t}\bigl[p(S(t),u^*_F(S(t)))+\omega\,S(t)\bigr]\,dt.
\]
Since the firm's policy is independent of~$\omega$, the skill trajectory
$S(t)$ and output path $p(t)$ are fixed.  Write
$W_\omega = W_0 + \omega\!\int_0^\infty e^{-\delta_W t}S(t)\,dt$,
where $W_0$ is the welfare at $\omega=0$.
Evaluating the integral using the exponential skill path:
\[
  \int_0^\infty e^{-\delta_W t}S(t)\,dt
  = \frac{\hat S_F}{\delta_W} + \frac{\bar S-\hat S_F}{\delta_W+\hat\kappa_F}.
\]
Therefore
\[
  W_\omega - V_W^{\text{no-AI}}
  = \underbrace{\bigl(W_0 - \bar S/\delta_W\bigr)}_{\text{productivity gap}}
    \;+\; \omega\left[\frac{\hat S_F}{\delta_W}
    + \frac{\bar S-\hat S_F}{\delta_W+\hat\kappa_F}
    - \frac{\bar S}{\delta_W}\right].
\]
The bracketed term equals
$-(\bar S-\hat S_F)\hat\kappa_F/[\delta_W(\delta_W+\hat\kappa_F)]<0$
when $\hat S_F<\bar S$ (i.e.\ when AI is adopted).
The first term $W_0-\bar S/\delta_W$ is non-negative by the aligned case (since $\delta_F=\delta_W$ and $\omega=0$).  But the negative $\omega$-term grows without bound in~$\omega$,
so for sufficiently large~$\omega$ the sum is negative and the trap holds.
\end{proof}

\section{Extensions}\label{app:extensions}

\subsection{Robustness: General Functional Forms}\label{app:robustness}

We now show that steady-state loss is not dependent on the linear-quadratic structure of the main model. Consider a single worker with skill $S(t)\in[0,\bar S]$ and AI usage $u(t)\in[0,1]$. Flow output is $y(S,u)$ and skill evolves according to $\dot S(t) = H(S(t),u(t))$.

\begin{definition}[No-AI steady state]
A no-AI steady state is a skill level $\bar S>0$ such that $H(\bar S,0)=0$, and it is \emph{locally stable} if $H_S(\bar S,0)<0$.
\end{definition}

\begin{assumption}[Skill dynamics]\label{ass:skill}
The skill law of motion $H$ satisfies:
(i) $H(\bar S,0)=0$ and $H_S(\bar S,0)<0$ (stable no-AI steady state);
(ii) $H_u(\bar S,0)<0$: at $\bar S$, a marginal increase in AI usage reduces net skill accumulation;
(iii) $H$ is continuously differentiable in a neighborhood of $(\bar S,0)$.
\end{assumption}

\begin{assumption}[Production]\label{ass:prod}
The flow output function $y$ satisfies:
(i) $y$ is continuously differentiable in a neighborhood of $(\bar S,0)$;
(ii) $y_S(\bar S,0)>0$: higher skill raises output at the no-AI steady state;
(iii) $y_u(\bar S,0)>0$: introducing AI at $\bar S$ raises current output.
\end{assumption}

\begin{lemma}[Steady-state skill response]\label{lem:Sstar}
Fix a constant usage level $u\in[0,1]$. Under Assumption~\ref{ass:skill}, there exists $\varepsilon>0$ and a differentiable function $S^\ast:[0,\varepsilon)\to\mathbb{R}$ with $S^\ast(0)=\bar S$ and
\[
S^{\ast\prime}(0) = -\frac{H_u(\bar S,0)}{H_S(\bar S,0)} < 0.
\]
\end{lemma}

\begin{proof}
By the implicit function theorem applied to $H(S^\ast(u),u)=0$ at $(\bar S,0)$, using $H_S(\bar S,0)\neq 0$.
\end{proof}

\begin{proposition}[General local steady-state loss condition]\label{prop:general-trap}
Under Assumptions~\ref{ass:skill} and~\ref{ass:prod}, define long-run output $y^\ast(u):=y(S^\ast(u),u)$. Then
\[
\frac{d y^\ast(u)}{du}\bigg|_{u=0}
=
\underbrace{y_u(\bar S,0)}_{\text{direct AI gain }B_0}
-
\underbrace{y_S(\bar S,0)\,m}_{\text{long-run skill cost }C_0},
\]
where $m := -S^{\ast\prime}(0) > 0$. There is local steady-state loss if and only if $0 < B_0 < C_0$.
\end{proposition}

\begin{proof}
By the chain rule: $dy^\ast/du|_{u=0} = y_S(\bar S,0) S^{\ast\prime}(0) + y_u(\bar S,0) = B_0 - C_0$.
\end{proof}

\paragraph{Connection to the parametric model.}
In the learning-forgetting law $\dot S = \kappa[\bar S(1-u)-S]$, we have $H_S = -\kappa$ and $H_u = -\kappa\bar S$, so $m = \bar S$. With $y(S,u) = S + [\alpha+(\beta-1) S - \gamma u]u$, the loss condition $0 < B_0 < C_0$ becomes $0 < \alpha+(\beta-1)\bar S < \bar S$, matching the parametric results.

\subsection{Proof of Proposition~\ref{prop:stratification} (Skill Stratification)}\label{app:kcurve}

This appendix proves the  skill stratification result stated in the main text. The key condition is $(1-\beta+2\kappa a\bar S)\bar S > 2\gamma$ (equivalently, $(1-\beta)\bar S > 2\gamma$ when $\kappa\to 0$), which ensures that the unconstrained optimal policy spans a wide enough range to simultaneously prescribe full automation for low-skill workers and no AI for high-skill workers.

\begin{remark}[Stratification when $D<0$]\label{rem:D-negative}
The discriminant $D$ depends only on $\beta$, not $\alpha$, and equals zero at $\beta = 1 - \gamma(2\kappa+\delta)/(2\kappa\bar{S})$. Below this threshold, no stable interior steady state exists: the skill-dependent feedback is strong enough that workers are eventually pushed to one of the two corner steady states ($\hat{S}=\bar{S}$, $\hat{u}=0$) or ($\hat{S}=0$, $\hat{u}=1$), depending on initial skill. The five-region classification from the main text therefore applies only above this boundary; below it, the long-run outcome is $S_0$-dependent (Figure~\ref{fig:extended-region-map}).
\end{remark}

\begin{figure}[ht]
\centering
\includegraphics[width=0.75\linewidth]{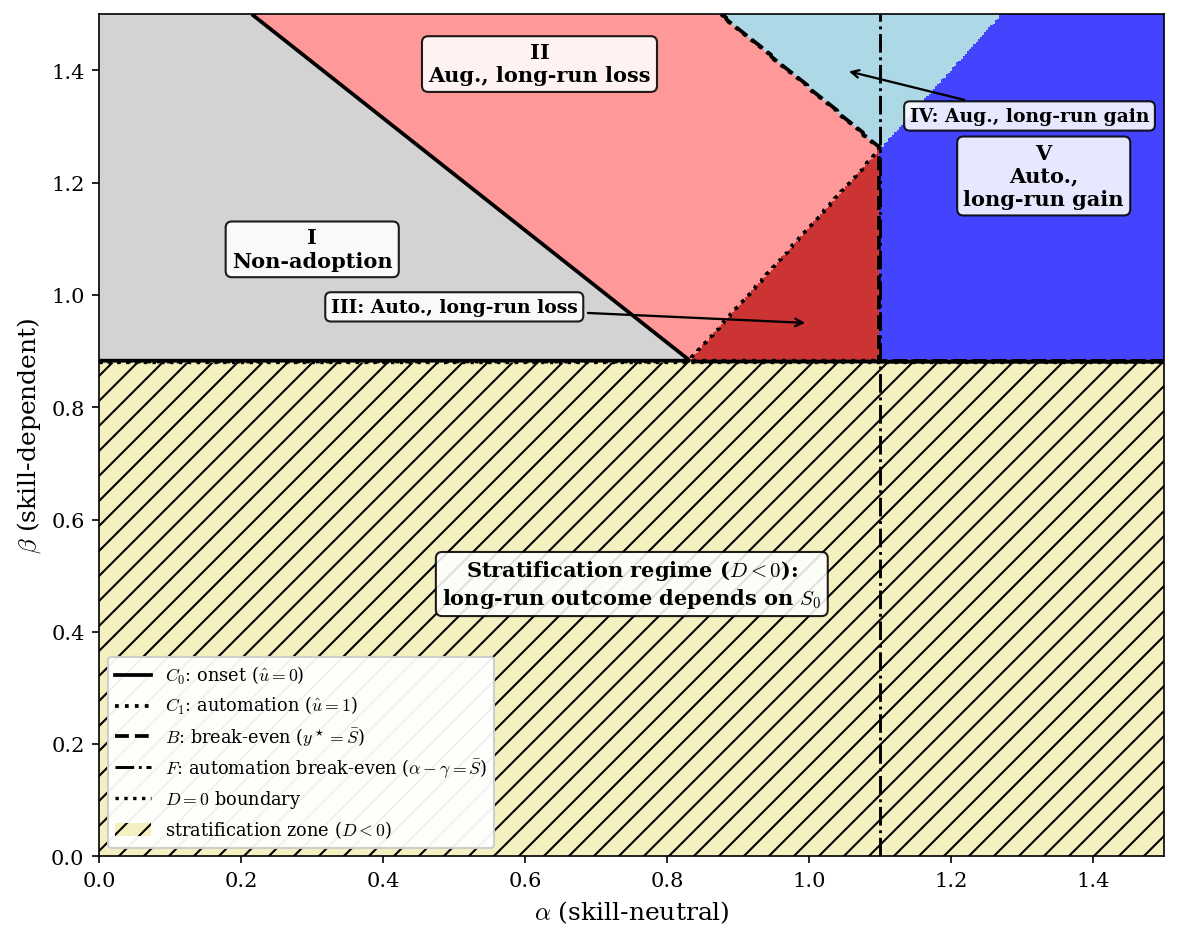}
\caption{\textbf{Extended region map with skill substitution.}
The map extends Figure~\ref{fig:nomathheatmap} to $\beta<1$. The hatched region marks the stratification regime, where no stable interior steady state exists and the long-run outcome depends on initial skill $S_0$. Parameters: $\kappa=0.25$, $\gamma=0.10$, $\delta=0.1$, $\bar S=1$.}
\label{fig:extended-region-map}
\end{figure}

\begin{figure}[ht]
\centering
\includegraphics[width=0.75\linewidth]{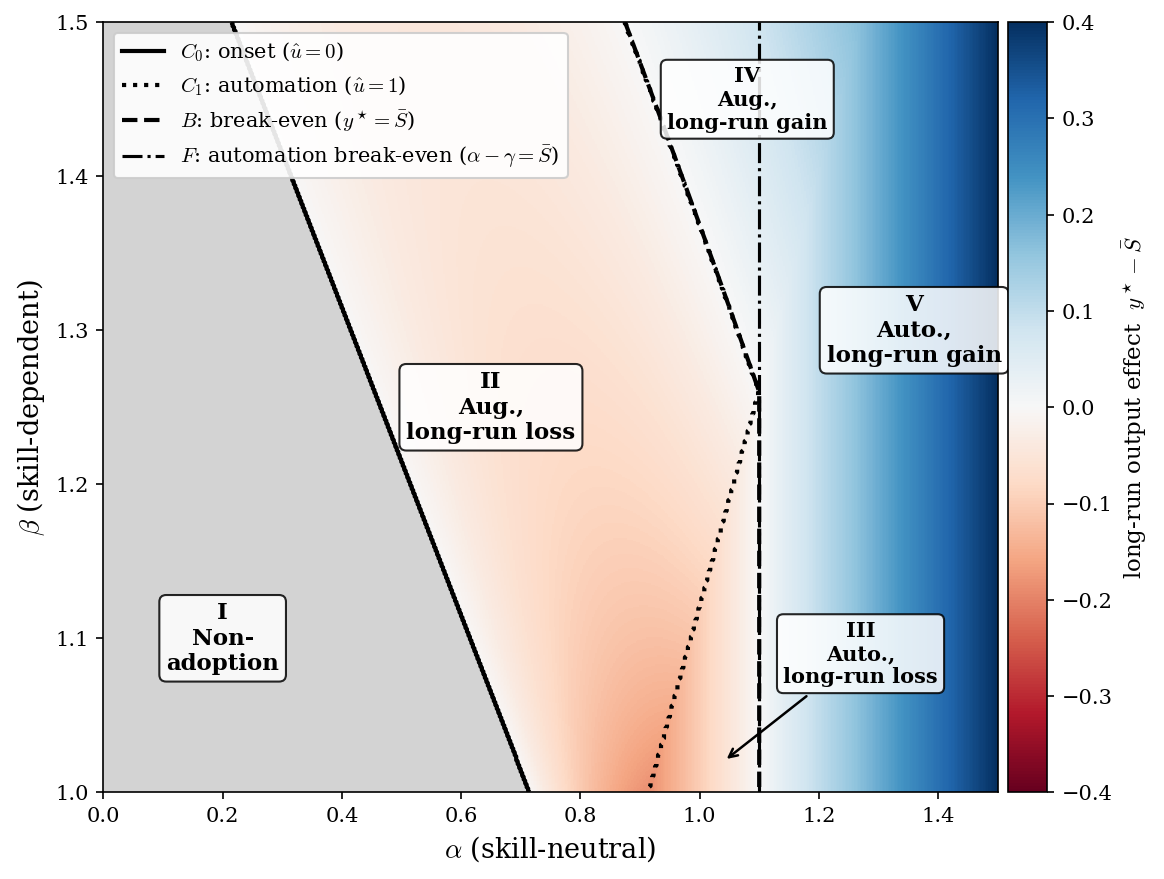}
\caption{\textbf{Magnitude of the long-run output effect.}
Shading shows $\Delta y^{ss}=y^\star-\bar S$ over the $(\alpha,\beta)$ space: red indicates long-run loss, blue indicates long-run gain, and gray indicates non-adoption. Boundaries match Figure~\ref{fig:nomathheatmap}. Parameters: $\gamma=0.10$, $\kappa=0.25$, $\delta=0.1$, $\bar S=1$.}
\label{fig:region-magnitude}
\end{figure}

\begin{proof}[Proof of Proposition~\ref{prop:stratification}]
Let $S_B := (1-u_0)/u_1$ denote the skill level at which the unconstrained policy $\tilde u(S)=u_0+u_1S$ reaches $u=1$, and $S_A := -u_0/u_1$ the level at which it reaches $u=0$; these coincide with the switching thresholds of Section~\ref{app:switching}. The hypotheses $\tilde u(0)=u_0>1$ and $\tilde u(\bar S)<0$ with $u_1<0$ give $0<S_B<S_A<\bar S$.

The condition $u_0>1$ means the unconstrained policy prescribes usage above the upper bound for low-skill workers: $\bar u=1$ for $S\le S_B$. Under full automation the dynamics reduce to $\dot S = -\kappa S$, so skill goes to zero. By symmetry, $\bar u=0$ for $S\ge S_A$, and $\dot S = \kappa(\bar S - S)>0$ drives skill toward $\bar S$.

In the region $S_B<S<S_A$, $\dot S = \kappa\bar S(1-u_0-u_1 S) - \kappa S$, with $\partial\dot S/\partial S = -\kappa(1+u_1\bar S)$. Since $u_0>1$ and $u_0+u_1\bar S<0$ imply $|u_1|\bar S>1$, we have $1+u_1\bar S<0$, so $\dot S$ is increasing in $S$. Setting $\dot S=0$ gives the unstable equilibrium $S_{\mathrm{eq}} = \bar S(1-u_0)/(1+u_1\bar S)$, and $u_0<|u_1|\bar S$ ensures $S_{\mathrm{eq}}\in(S_B,S_A)$. At $S=S_B$, $\dot S = -\kappa S_B<0$; since $\dot S$ is increasing in $S$, $\dot S<0$ on $(S_B,S_{\mathrm{eq}})$ and $\dot S>0$ on $(S_{\mathrm{eq}},S_A)$. Therefore, workers with $S_0<S_{\mathrm{eq}}$ lose skill, cross $S_B$, enter the full-automation regime, and converge to $\hat S=0$, while workers with $S_0>S_{\mathrm{eq}}$ drift up, cross $S_A$, and reach full potential at $\bar S$.
\end{proof}

The changes in the trajectories come from these regime transitions. Near $S_{\mathrm{eq}}$, $\dot S \approx 0$, so trajectories move slowly before entering either basin. Once a worker crosses $S_B$ or $S_A$, the dynamics produce kinks in the skill paths at those thresholds.

When $0<u_0\le 1$, the argument above shows $\dot S > 0$ everywhere below the zero-usage cutoff, so all workers eventually exit the adoption region and converge to a common steady state. The population split therefore requires that the optimal policy be aggressive enough to hit the upper constraint.

The condition $u_0+u_1\bar S< 0$ with $u_1<0$ (which follows from $\beta<1$) requires $|u_1|\bar S > u_0 > 1$, placing a lower bound on $|u_1|$ and hence on $1-\beta$. In the limit $\kappa\to 0$, the unconstrained policy reduces to $u^*(S) = (\alpha+(\beta-1) S)/(2\gamma)$, so the conditions $u^*(0)>1$ and $u^*(\bar S)< 0$ become $\alpha > 2\gamma$ and $\alpha < (1-\beta)\bar S$. These can be  satisfied if and only if $(1-\beta)\bar S > 2\gamma$. At $\gamma=1$, $\bar S = 1$, this would require $\beta<-1$, outside any economically meaningful range. Because $\bar S$ is a free parameter, however, the condition holds for $\beta\in(0,1)$ when $\bar S > 2\gamma/(1-\beta)$. For finite $\kappa$, substituting $u_0=(\alpha-\kappa b\bar S)/(2\gamma)$ and $u_1=(\beta-1-2\kappa a\bar S)/(2\gamma)$, the conditions $u_0>1$ and $u_0+u_1\bar S< 0$ become $\alpha\in(2\gamma+\kappa b\bar S,\;(1-\beta)\bar S+\kappa b\bar S+2\kappa a\bar S^2)$. The value-function coefficients $a,b$ depend on the model parameters through the coefficient-matching conditions~\eqref{eq:coeff-a}--\eqref{eq:coeff-b}; Figure~\ref{fig:k-curve} uses a parameter point inside this feasibility region with $D>0$.

Because usage decreases in skill when $\beta<1$, low-skill workers adopt more aggressively, and the resulting atrophy widens the skill gap during the transition. When $(1-\beta+2\kappa a\bar S)\bar S \le 2\gamma$, this divergence is temporary and all workers converge to a common steady state. When $(1-\beta+2\kappa a\bar S)\bar S > 2\gamma$, the divergence is permanent. The condition is easier to satisfy when workers span a wide range of potential (large $\bar S$) or when AI substitutes heavily for skill (small $\beta$).

\subsection{Proof of Proposition~\ref{prop:basin-switch}}\label{app:basin-switch}

\begin{proof}
Under the stratification condition ($\beta<1$, $(1-\beta+2\kappa a\bar S)\bar S>2\gamma$),
the optimal policy satisfies $u_0>1$ and $1+u_1\bar S<0$.
The threshold is
\[
  S_{\mathrm{eq}} = \frac{\bar S(1-u_0)}{1+u_1\bar S}.
\]
Both numerator and denominator are negative (since $u_0>1$ and $1+u_1\bar S<0$),
so $S_{\mathrm{eq}}>0$.

\textbf{Comparative statics in $\delta$.}\quad
The threshold is
$S_{\mathrm{eq}} = \bar S(1-u_0)/(1+u_1\bar S)$.
Implicitly differentiating the equilibrium condition
$u_0 + u_1 S_{\mathrm{eq}} = 1 - S_{\mathrm{eq}}/\bar S$
with respect to~$\delta$:
\[
  \frac{\partial S_{\mathrm{eq}}}{\partial\delta}
  = \frac{-\bigl(\frac{\partial u_0}{\partial\delta}
    + \frac{\partial u_1}{\partial\delta}\,S_{\mathrm{eq}}\bigr)}
    {u_1 + 1/\bar S}.
\]
The denominator satisfies $u_1 + 1/\bar S < 0$ under stratification, and the numerator is $-\partial u^*/\partial\delta$ evaluated at $S_{\mathrm{eq}}$, so $\partial S_{\mathrm{eq}}/\partial\delta$ has the sign of $\partial u^*(S_{\mathrm{eq}})/\partial\delta$. The sign of $b_\delta$ from Proposition~\ref{prop:overuse} is not available here, because that argument uses the stable case $\hat\kappa>0$ and stratification requires $\hat\kappa<0$, so we sign $\partial u^*(S_{\mathrm{eq}})/\partial\delta$ directly. The envelope identity~\eqref{eq:envelope} holds regardless of the sign of~$\hat\kappa$. Differentiating it in $S$ and evaluating at $S=S_{\mathrm{eq}}$, where $\dot S = 0$, gives
\[
  V_{S\delta}(S_{\mathrm{eq}}) = -\frac{V'(S_{\mathrm{eq}})}{\delta+\hat\kappa}.
\]
Since $\delta+2\hat\kappa = \sqrt{D}/\gamma > 0$ and $\hat\kappa<0$, we have $\delta+\hat\kappa>0$, and $V'(S_{\mathrm{eq}}) = 2aS_{\mathrm{eq}}+b > 0$ because $a>0$, $b>0$, and $S_{\mathrm{eq}}>0$. Hence $V_{S\delta}(S_{\mathrm{eq}})<0$, so $\partial u^*(S_{\mathrm{eq}})/\partial\delta = -\kappa\bar S\,V_{S\delta}(S_{\mathrm{eq}})/(2\gamma) > 0$, and $\partial S_{\mathrm{eq}}/\partial\delta > 0$.

\paragraph{Comparative statics in~$\omega$.}
From the proof of Proposition~\ref{prop:outside-options}, $a$ is independent of~$\omega$ (the $S^2$ coefficient-matching equation is unchanged), so $\partial u_1/\partial\omega=0$.
Since $\partial b/\partial\omega = 2\gamma/(\gamma\delta+\sqrt{D})>0$,
\[
  \frac{\partial u_0}{\partial\omega}
  = \frac{-\kappa\bar S}{2\gamma}\,\frac{\partial b}{\partial\omega} < 0.
\]
Therefore
\[
  \frac{\partial S_{\mathrm{eq}}}{\partial\omega}
  = \underbrace{\frac{-\bar S}{1+u_1\bar S}}_{>0}
    \cdot\underbrace{\frac{\partial u_0}{\partial\omega}}_{<0}
  < 0.
\]

\paragraph{Consequences.} \textit{Part 1.} Since $S_{\mathrm{eq}}$ is strictly increasing in~$\delta$ and $\delta_F>\delta_W$,
we have $S_{\mathrm{eq}}(\delta_F)>S_{\mathrm{eq}}(\delta_W)$.
Any worker with initial skill $S_0\in(S_{\mathrm{eq}}(\delta_W),\,S_{\mathrm{eq}}(\delta_F))$
lies above the threshold under the worker's policy (converging to~$\bar S$ by
Proposition~\ref{prop:stratification}) and below it under the firm's (converging to~$0$).

\medskip\noindent \textit{Part 2.} The firm ignores $\omega$ and uses $S_{\mathrm{eq}}(\omega\!=\!0)$.
The worker, who values skill at rate~$\omega$, would adopt less AI ($u_0$ lower),
producing $S_{\mathrm{eq}}(\omega)<S_{\mathrm{eq}}(0)$.
Any worker with $S_0\in(S_{\mathrm{eq}}(\omega),\,S_{\mathrm{eq}}(0))$
converges to~$\bar S$ under the worker's own policy but to~$0$ under the firm's.
\end{proof}

\subsection{Regime Switching Thresholds and Time to Entry}\label{app:switching}

As skill evolves under the optimal policy, a worker may cross the thresholds at which the feasible policy changes regime. From $u^*(S) = u_0 + u_1 S$, define the adoption threshold $S_A$ (where $u^* = 0$) and the automation threshold $S_B$ (where $u^* = 1$):
\begin{align}
S_A &= -\frac{u_0}{u_1} = \frac{\kappa b\bar S - \alpha}{(\beta-1) - 2\kappa a\bar S}, \label{eq:SA}\\[4pt]
S_B &= \frac{1 - u_0}{u_1} = \frac{2\gamma + \kappa b\bar S - \alpha}{(\beta-1) - 2\kappa a\bar S}. \label{eq:SB}
\end{align}
When $\beta > 1$, returns to skill are high and $u_1 > 0$, so $S_A < S_B$: workers below $S_A$ optimally avoid AI, while those above $S_B$ use more AI. When $\beta < 1$, returns to skill are lower and $u_1 < 0$, $S_A > S_B$, meaning low-skill workers use AI and high-skill workers abstain.

\paragraph{Time to entry from the no-AI regime.}
Workers starting at $S_{0} < S_A$ (when $\beta > 1$) or $S_0 > S_A$ (when $\beta < 1$) initially practice without AI. Under $u = 0$, skill follows $\dot S = \kappa(\bar S - S)$, giving $S(t) = \bar S - (\bar S - S_0)e^{-\kappa t}$. The time to reach the adoption threshold is
\begin{equation}\label{eq:tau-adopt}
\tau_A = \frac{1}{\kappa}\ln\!\left(\frac{\bar S - S_0}{\bar S - S_A}\right),
\end{equation}
provided $S_0 < S_A < \bar S$ (for $\beta > 1$). If $S_A \ge \bar S$, the worker abstains.

\paragraph{Time to entry from the full-automation regime.}
A worker who starts in $u = 1$ (skill below $S_B$ when $\beta < 1$) follows $\dot S = -\kappa S$, so $S(t) = S_0 e^{-\kappa t}$. Skill decays to zero. When $\beta > 1$ and a worker starts above $S_B$ with $u = 1$, skill decays until it reaches $S_B$:
\begin{equation}\label{eq:tau-revert}
\tau_B = \frac{1}{\kappa}\ln\!\left(\frac{S_0}{S_B}\right),
\end{equation}
after which the worker reverts to augmentation.

\paragraph{Dynamics in the interior regime.}
Substituting the interior policy \(u^*(S)=u_0+u_1S\) into the skill law gives
\[
\dot S
=
-\hat\kappa(S-S_{\mathrm{int}}),
\qquad
\hat\kappa \equiv \kappa(1+u_1\bar S),
\qquad
S_{\mathrm{int}}
\equiv
\frac{\bar S(1-u_0)}{1+u_1\bar S}.
\]
Thus, as long as the path remains interior,
\[
S(t)-S_{\mathrm{int}}
=
(S_0-S_{\mathrm{int}})e^{-\hat\kappa t}.
\]
If \(\hat\kappa>0\), the interior crossing is stable and skill converges to \(S_{\mathrm{int}}\) at rate \(\hat\kappa\), with time constant \(1/\hat\kappa\). This always holds for the complementary case \(\beta>1\), since then \(u_1>0\). If \(\hat\kappa<0\), the interior crossing is unstable: for any \(S_0\neq S_{\mathrm{int}}\), deviations grow at rate \(|\hat\kappa|\) until the path reaches a usage boundary. Under the stratification conditions, initial skill below \(S_{\mathrm{int}}\) leads to the full-automation region \(u=1\), while initial skill above \(S_{\mathrm{int}}\) leads to the no-automation region \(u=0\).

\end{document}